\documentclass[11pt]{article}

\usepackage[margin=1in]{geometry}

\usepackage[
  colorlinks,
  linkcolor = blue,
  citecolor = blue,
  urlcolor = blue]{hyperref}

\usepackage{amsmath,amssymb,amsthm}

\usepackage{tikz}
\usetikzlibrary{arrows}
\usetikzlibrary{calc,positioning}

\usepackage{thm-restate}

\usepackage{colonequals}

\usepackage[font=small]{caption}

\newcommand{\C}{\mathbb{C}}

\newcommand{\U}{\mathrm{U}}

\newcommand{\ket}[1]{|#1\rangle}
\newcommand{\bra}[1]{\langle#1|}
\newcommand{\braket}[2]{\langle#1|#2\rangle}
\newcommand{\ketbra}[2]{|#1\rangle\langle#2|}
\newcommand{\ketbras}[1]{|#1\rangle\langle#1|}
\newcommand{\set}[1]{\left\{#1\right\}}
\newcommand{\abs}[1]{|#1|}            
\newcommand{\Abs}[1]{\left|#1\right|} 
\newcommand{\norm}[1]{\left\|#1\right\|}

\newcommand{\mc}[1]{\mathcal{#1}}
\newcommand{\mx}[1]{\begin{pmatrix}#1\end{pmatrix}}

\newcommand{\tp}{^{\textrm{T}}}
\newcommand{\ct}{^{\dagger}}
\newcommand{\x}{\otimes}
\newcommand{\vp}{\varphi}
\newcommand{\ve}{\varepsilon}
\newcommand{\id}{I}
\newcommand{\bip} [2]{\C^{#1}\x\C^{#2}}    
\newcommand{\bips}[1]{\bip{#1}{#1}}        

\newcommand{\partitle}[1]{\vspace{-10pt}\subsubsection*{\normalfont\emph{#1}}}
\newcommand\error{\operatorname{error}}
\newcommand{\defeq}{\colonequals}

\renewcommand\Re{\operatorname{Re}}
\renewcommand\Im{\operatorname{Im}}
\renewcommand\vec{\operatorname{vec}}
\DeclareMathOperator{\tr}{Tr}

\DeclareMathOperator{\supp}{supp}

\DeclareMathOperator{\rows}{rows}
\DeclareMathOperator{\cols}{cols}

\DeclareMathOperator{\Lin}{L}
\DeclareMathOperator{\Pos}{Pos}

\newtheoremstyle{noCaption}
{\topsep}
{\topsep}
{\itshape}
{}
{}
{}
{0pt}
{}%

\newtheorem{theorem}{Theorem}
\newtheorem*{theorem*}{Theorem}
\newtheorem{lemma}{Lemma}
\newtheorem{definition}{Definition}

\newtheorem*{claim*}{Claim}
\newtheorem*{fact*}{Fact}
\newtheorem{corollary}{Corollary}
\newtheorem*{example*}{Example}

\theoremstyle{definition}

\theoremstyle{noCaption}
\newtheorem*{problem*}{}

\title{A framework for bounding nonlocality of state discrimination}
\author{Andrew M.~Childs, Debbie Leung, Laura Man\v{c}inska, and Maris Ozols \\[2pt] \normalsize
Department of Combinatorics \& Optimization \\ \normalsize
and Institute for Quantum Computing \\ \normalsize
University of Waterloo}
\date{}

\begin{document}
\maketitle

\begin{abstract}
We consider the class of protocols that can be implemented by local quantum operations and classical communication (LOCC) between two parties. In particular, we focus on the task of discriminating a known set of quantum states by LOCC. Building on the work in the paper \emph{Quantum nonlocality without entanglement} \cite{IBM}, we provide a framework for bounding the amount of nonlocality in a given set of bipartite quantum states in terms of a lower bound on the probability of error in any LOCC discrimination protocol. We apply our framework to an orthonormal product basis known as the domino states and obtain an alternative and simplified proof that quantifies its nonlocality. We generalize this result for similar bases in larger dimensions, as well as the ``rotated'' domino states, resolving a long-standing open question \cite{IBM}.
\end{abstract}

\tableofcontents

\newpage

\section{Introduction}

The 1999 paper \emph{Quantum nonlocality without entanglement} \cite{IBM} exhibits an orthonormal basis $S \subset \bip{3}{3}$ of product states, known as domino states, shared between two separated parties. When the parties are restricted to perform only local quantum operations and classical communication (LOCC), it is impossible to discriminate the domino states arbitrarily well \cite{IBM}. In such cases we say that perfect discrimination cannot be achieved with \emph{asymptotic} LOCC. Moreover, \cite{IBM} also quantifies the extent to which any LOCC protocol falls short of perfect discrimination of the domino states.

This result spurred interest in state discrimination with LOCC. Several alternative proofs \cite{Walgate, Groisman, Cohen07} of the impossibility of perfect LOCC discrimination of the domino states were given along with many other results concerning perfect state discrimination (e.g., \cite{BDMSST, Walgate-Multi, Ghosh-Bell, Groisman, Virmani, ChenYang-Multipartite, ChenYang-Entangled, Walgate, DMSST, ChenLi-Criterion, Horodecki, HilleryMimih, Fan, Ghosh-MaxEnt, Chefles, ChenLi-ProdBasis, JiCaoYing, Watrous, Nathanson, NisetCerf, DuanUB, Feng3, DuanSep, DuanXinYing}). However, the problem of asymptotic LOCC state discrimination has not received much attention since the initial study of nonlocality without entanglement \cite{IBM}.

The main motivation for our work is to better understand the phenomenon of quantum nonlocality without entanglement. More concretely, our goals are to
\begin{itemize}
  \item simplify the original proof,
  \item render the technique applicable to a wider class of sets of bipartite states,
  \item exhibit new classes of product bases that cannot be \emph{asymptotically} (as opposed to just perfectly) discriminated with LOCC,
  \item pin down where exactly the difference between LOCC and separable operations lies, and
  \item investigate the possibility of larger gaps between the sets of LOCC and separable operations.
\end{itemize}

In particular, we seek to exhibit quantitative gaps between the classes of LOCC and separable operations. Separable operations often serve as a relaxation of LOCC operations and such gaps show how imprecise this relaxation can be. The rationale behind this relaxation is that separable operations have a clean mathematical description whereas LOCC operations can be much harder to understand.

There is also an operational motivation to quantify the difference between separable measurements and those implemented by asymptotic LOCC: the former are precisely the measurements that cannot generate entangled states, while the latter are those that do not require entanglement to implement \cite{IBM, Koashi07, Koashi09}.  Thus, a separable measurement that cannot be implemented by asymptotic LOCC uses entanglement irreversibly.

\subsection*{Our contributions}

In this paper, we develop a framework for obtaining quantitative results on the hardness of quantum state discrimination by LOCC. More precisely, we provide a method for proving a lower bound on the error probability of any LOCC measurement for discriminating states from a given set $S$.

Our first main contribution (Theorem~\ref{thm:eta}) is that
any LOCC measurement for discriminating states from a set $S$ errs with probability $p_{\error} \geq \frac{2}{27} \frac{\eta^2}{|S|^5}$, where
$\eta$ is a constant that depends on $S$ (see Definition~\ref{def:eta}). Intuitively, $\eta$ measures the nonlocality of $S$.

Our second main contribution is a systematic method for bounding the nonlocality constant $\eta$ for a large class of product bases. Together with the above theorem, this lets us quantify the hardness of LOCC discrimination for the following bases of product states:

\begin{enumerate}
  \item \emph{domino states}, the original set of nine states in $3 \times 3$ dimensions first considered in~\cite{IBM}, have $p_{\error} \geq 1.9 \times 10^{-8}$;
  \item \emph{domino-type states}, a generalization of domino states to higher dimensions corresponding to tilings of a rectangular $d_A \times d_B$ grid by tiles of size at most two, have $p_{\error} \geq 1/(216 D^2 d_A^5 d_B^5)$, where $D$ is a property of the tiling that we call ``diameter'';
  \item \emph{$\theta$-rotated domino states}, a $1$-parameter family that includes the domino states and the standard basis as extreme cases, have $p_{\error} \geq 2.4 \times 10^{-11} \sin^2 2 \theta$ (determining whether these states can be discriminated perfectly by LOCC and finding a lower bound on the probability of error were left as open problems in~\cite{IBM}).
\end{enumerate}

The rest of the paper is organized as follows. In Section~\ref{sec:Background} we introduce notation, give background on LOCC measurements and state discrimination, and summarize related prior work. In Section~\ref{sec:Framework} we introduce our general framework for lower bounding the error probability of LOCC measurements, and in Section~\ref{sec:LowerBd} we prove Theorem~\ref{thm:eta}. In Section~\ref{sec:Tilings} we consider the case where $S$ is a product basis and propose a method for bounding the nonlocality constant $\eta$ by another quantity that we call ``rigidity.'' Our approach is based on a description of sets of bipartite states in terms of tilings. In Section~\ref{sec:Dominoes} we define the three classes of states mentioned above and prove a bound on the rigidity of the domino states; bounds on the rigidity of the domino-type states and the rotated domino states appear in Appendices~\ref{apx:DimBox} and~\ref{apx:RotBox}, respectively. Finally, we discuss limitations of our framework in Section~\ref{sec:Limitations} and conclude with a discussion of open problems in Section~\ref{sec:Conclusions}.

\section{Background} \label{sec:Background}

\subsection{Notation}

The following notation is used in this paper. Let $\Lin(\C^n,\C^m)$ be the set of all linear operators from $\C^n$ to $\C^m$ and let $\Lin(\C^n) \defeq \Lin(\C^n,\C^n)$. Next, let $\Pos(\C^n) \subset \Lin(\C^n)$ be the set of all positive semidefinite operators on $\C^n$. Let $\norm{M}_{\max} \defeq \max_{ij} \abs{M_{ij}}$ denote the largest entry of $M \in \Lin(\C^n)$ in absolute value. Finally, for any natural number $n$, let $[n] \defeq \set{1,\dotsc, n}$ and let $I_n$ be the $n \times n$ identity matrix.

\subsection{Separable and LOCC measurements}

A $k$-outcome \emph{POVM measurement} (or simply a measurement) on an $n$-dimensional state space is a set of operators $\set{E_1, \dotsc, E_k} \subset \Pos(\C^n)$ such that $\sum_{i=1}^k E_i = \id_n$. The operators $E_i$ are called \emph{POVM elements}. The probability of obtaining outcome $i$ upon measuring state $\rho$ is $\tr(E_i \rho)$. 

When it is necessary to keep track of the post-measurement state, it is more convenient to use a \emph{non-destructive} measurement. Such a measurement is specified by a set of \emph{measurement operators} $\set{M_1, \dotsc, M_k} \subset \Lin(\C^n,\C^m)$ for some finite $m$ where $\sum_{i=1}^k M_i\ct M_i = \id_n$. The probability of obtaining outcome $i$ upon measuring state $\rho$ is $\tr(M_i\ct M_i \rho)$ 
and the $m$-dimensional post-measurement state is $M_i \rho M_i\ct/\tr(M_i\ct M_i \rho)$.

Note that a non-destructive measurement followed by discarding the post-measurement state corresponds to a POVM measurement with elements $E_i = M_i\ct M_i$.

\subsubsection{Separable measurements}

\begin{definition}
A measurement $\mc{E} = \set{E_1, \dotsc, E_k}$ on a bipartite state space $\bip{d_A}{d_B}$ is \emph{separable} if all POVM elements $E_i$ are separable, i.e.,
\begin{equation}
  E_i = \sum_j E^A_j \x E^B_j
\end{equation}
for some $E^A_j \in \Pos(\C^{d_A})$ and $E^B_j\in\Pos(\C^{d_B})$.
\end{definition}

Note that the above definition is equivalent to saying that $\mc{M}$ is obtained from a measurement with product POVM elements, followed by classical post-processing (coarse graining).

\subsubsection{LOCC measurements}

Informally, a bipartite $n$-outcome LOCC measurement $\mc{E}$ consists of the two parties taking finitely many turns (called \emph{rounds}) of applying adaptive non-destructive measurements to their state spaces and exchanging classical messages. This is followed by coarse graining all measurement records into $n$ bins, each corresponding to one of the $n$ outcomes of $\mc{E}$.

Let us describe such a protocol $\mc{E}$ more formally, adopting notation similar to that of \cite{IBM}. Let $\Lambda$ denote the empty string, corresponding to no message being sent. The protocol begins when one of the parties, say Alice, applies a non-destructive measurement
\begin{equation}
  \mc{A}(\Lambda) = \set{A_1(\Lambda), \dotsc, A_{k(\Lambda)}(\Lambda)}
\end{equation}
to her state space and communicates the round 1 measurement outcome $m_1 \in [k(\Lambda)]$ to Bob. Then, depending on the value of $m_1$ received, Bob applies a non-destructive measurement
\begin{equation}
  \mc{B}(m_1) = \set{B_1(m_1), \dotsc, B_{k(m_1)}(m_1)}
\end{equation} 
to his state space and communicates the round 2 measurement outcome $m_2 \in [k(m_1)]$ to Alice. The protocol proceeds with the two parties taking finitely many alternating turns of a similar form, where the non-destructive measurement applied at round $t$ depends on the \emph{measurement record} $m = (m_1, \dotsc ,m_{t-1})$ accumulated during the previous rounds.

Let $m$ be the measurement record after the execution of the first $t$ rounds of the protocol. Then the measurement operator that Alice and Bob have effectively implemented is a product operator $A_m \x B_m$, where\footnote{Here we assume for simplicity that $t$ is even; in the odd case the operators $A_m$ and $B_m$ can be defined similarly.}
\begin{align}
  A_m &\defeq A_{m_{t-1}}(m_1, \dotsc, m_{t-2}) \dotso
          A_{m_3}(m_1, m_2) A_{m_1}(\Lambda), \\
  B_m &\defeq B_{m_t}(m_1, \dotsc, m_{t-1}) \dotso
          B_{m_4}(m_1, m_2, m_3) B_{m_2}(m_1).
\end{align}

Alice and Bob may choose to terminate the protocol depending on the measurement record obtained. At this point they must output one of the $n$ outcomes of the LOCC measurement $\mc{E}$ that they are implementing. If $L(k)$ is the set of all terminating measurement records corresponding to outcome $k \in [n]$, then the $k$th POVM element of $\mc{E}$ is given by
\begin{equation}
  E_k \defeq \!\! \sum_{m \in L(k)} A_m\ct A_m \x B_m\ct B_m.
\end{equation}
Since each $E_k$ is separable, any LOCC measurement is separable.

\subsubsection{Finite and asymptotic LOCC}

We consider two scenarios: when a measurement can be performed in a finite number of rounds or asymptotically.

\begin{definition}
We say that a measurement $\mc{E}$ can be \emph{implemented by (finite) LOCC} if there exists a finite-round LOCC protocol that, for any input state, produces the same distribution of measurement outcomes as $\mc{E}$.
\end{definition} 

\begin{definition}
We say that a measurement $\mc{E}$ can be \emph{implemented by asymptotic LOCC} if there exists a sequence $\mc{P}_1, \mc{P}_2, \dotsc$ of finite-round LOCC protocols whose output distributions converge to that of $\mc{E}$.
\end{definition}

The exact implementation scenario is not practical since any real-world device is susceptible to errors due to imperfections in implementation. However, proving that a certain task cannot be performed asymptotically is considerably harder than showing that it cannot be done (exactly) by any finite LOCC protocol.

\subsubsection{LOCC protocol as a tree}

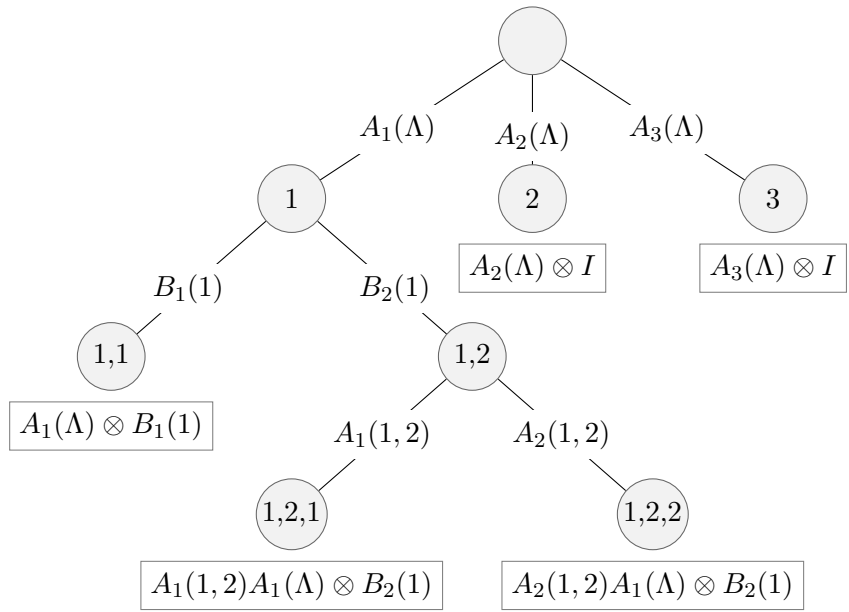
\begin{figure}[!ht]
\centering


\begin{tikzpicture}
[vx/.style = {circle, draw = black!60, fill = gray!10, minimum size = 9.0mm, inner sep = 1pt},
 every label/.style = {draw = gray, label distance = 4pt},
 edg/.style = {fill = white}]

\tikzstyle{level 1} = [level distance = 2.1cm, sibling distance = 3.2cm]
\tikzstyle{level 2} = [level distance = 2.1cm, sibling distance = 4.8cm]

\node[vx](root){} child foreach \i in {1,2,3} {
  \ifnum \i = 1 {
    node[vx] {\i}
    child foreach \j in {1,2}{
    \ifnum \j = 2 {
      node[vx]{\i,\j} 
      child foreach \k in {1,2}{
      node[vx, label = below: {$A_{\k}(1,2) A_1(\Lambda) \otimes B_2(1)$}]    
                              {\i,\j,\k}
      edge from parent node[edg]{$A_{\k}(\i,\j)$}}}
    \else{node[vx, label = below: {$A_{\i}(\Lambda) \otimes B_1(\i)$}]{\i,\j}}
    \fi
    edge from parent node[edg] {$B_{\j}(1)$}}}
   \else{node[vx, label = below: {$A_{\i}(\Lambda) \otimes I$}] {\i}}
   \fi
  edge from parent node[edg]{$A_{\i}(\Lambda)$}
};


\end{tikzpicture}

\caption{Tree structure of a three-outcome LOCC measurement. In round one Alice performs a three-outcome non-destructive measurement $\mc{A}(\Lambda)$; in round two, upon receiving message ``1'', Bob performs a two-outcome non-destructive measurement $\mc{B}(1)$ and upon receiving message ``2'' or ``3'' he terminates the protocol; in round three, upon receiving message ``1'', Alice terminates the protocol and, upon receiving message ``2'', she performs a two-outcome non-destructive measurement $\mc{A}(1,2)$. All nodes are labeled by the accumulated measurement record. The corresponding measurement operator is given below each leaf.}
\label{fig:Tree}
\end{figure}

We represent an LOCC measurement protocol as a tree (see Figure~\ref{fig:Tree}). The protocol begins at the root and proceeds downward along the edges. Each edge represents a certain measurement outcome obtained at its parent node, and \emph{leaves} are the nodes where the protocol terminates. The set of all leaves is partitioned into subsets, each corresponding to an outcome of the LOCC measurement being implemented.

A path from the root to a leaf is called a \emph{branch}. There is a one-to-one correspondence between the branches and the possible courses of execution of the LOCC protocol. Likewise, there is a one-to-one correspondence between the nodes of the tree and the accumulated measurement records.

The \emph{measurement at node $u$} is the measurement performed by the acting party once the protocol has reached node $u$. In contrast, the \emph{measurement operator corresponding to node $u$} is the measurement operator that has been implemented upon reaching node $u$. For example, consider the node $(1,2)$.  The measurement at node $(1,2)$ is given by the POVM $\{A_1(1,2), A_2(1,2)\}$, whereas the measurement operator corresponding to the node $(1,2)$ is given by $A_1(\Lambda) \otimes B_2(1)$.  As another example, the measurement operators corresponding to the leaves are exactly the measurement operators of the LOCC protocol prior to coarse graining.

\subsection{Bipartite state discrimination problem}

The goal of this paper is to investigate the limitations of two-party LOCC protocols for the task of bipartite quantum state discrimination, which is as follows:

\begin{problem*}
Let $S = \set{\ket{\psi_1}, \dotsc, \ket{\psi_n}} \subset \bip{d_A}{d_B}$ be a known set of quantum states. Suppose that $k \in [n]$ is selected uniformly at random and Alice and Bob are given the corresponding parts of state $\ket{\psi_k} \in S$. Their task is to determine the index $k$ by performing a measurement on this state.
\end{problem*}

A case of special interest is when $S$ is an orthonormal product basis, i.e., each $\ket{\psi_i} = \ket{\alpha_i} \ket{\beta_i}$ for some orthonormal bases $\ket{\alpha_i}\in\C^{d_A}$ and $\ket{\beta_i}\in\C^{d_B}$. Such states can be perfectly discriminated by a separable measurement $\mc{E}$ with POVM elements
\begin{equation}
  E_i \defeq \ketbras{\alpha_i} \x \ketbras{\beta_i}.
\end{equation}
However, this measurement cannot always be implemented by finite~\cite{Walgate, Groisman} or even asymptotic LOCC~\cite{IBM}. In such cases we say that $S$ possesses nonlocality (without entanglement).

\subsection{Previous results}

The first example of an orthonormal product basis of bipartite quantum states that cannot be perfectly discriminated by (even asymptotic) LOCC was given in \cite{IBM}. This is a striking illustration of the difference between the power of LOCC and separable operations. Furthermore, \cite{IBM} quantifies the information deficit of any LOCC protocol for discriminating these states. This result has been a starting point for many other studies on state discrimination by LOCC, with the ultimate goal of understanding LOCC operations and how they differ from separable ones. We briefly describe some of the directions that have been explored. Unless otherwise stated, these results refer to the discrimination of \emph{pure} states with \emph{finite} LOCC.

First consider the problem of discriminating two states without any restrictions on their dimension. Surprisingly, any two orthogonal (possibly entangled) pure states can be perfectly discriminated by LOCC, even when they are held by more than two parties \cite{Walgate-Multi}. Furthermore, optimal discrimination of any two multipartite pure states can be achieved with LOCC both in the sense of minimum error probability \cite{Virmani} and unambiguous discrimination \cite{ChenYang-Multipartite, ChenYang-Entangled, JiCaoYing}. Recently this has been generalized to implementing an arbitrary POVM by LOCC in any $2$-dimensional subspace \cite{Croke}.

Many authors have considered the problem of perfect state discrimination by \emph{finite} LOCC. In particular, the case where one party holds a small-dimensional system is well understood. Reference \cite{Walgate} characterizes when a set of orthogonal (possibly entangled) states in $\bips{2}$ can be perfectly discriminated by LOCC. A similar characterization for sets of orthogonal \emph{product} states in $\bips{3}$ has been given by \cite{Feng3}. In addition, \cite{Walgate} characterizes when a set of orthogonal states in $\bip{2}{n}$ can be perfectly discriminated by LOCC when Alice performs the first nontrivial measurement. It is also known that $\theta$-rotated domino states cannot be perfectly discriminated by LOCC (unless $\theta = 0$) \cite{Groisman}. Furthermore, the original domino states have inspired a construction of $n$-partite $d$-dimensional product bases that cannot be perfectly discriminated with LOCC \cite{NisetCerf}. 

The role of entanglement in perfect state discrimination by finite LOCC has also been considered. It is not possible to perfectly discriminate more than two Bell states by LOCC~\cite{Ghosh-Bell}. In fact, the same is true for any set of more than $n$ maximally entangled states in $\bips{n}$ \cite{Nathanson}. Multipartite states from an orthonormal basis can be perfectly discriminated by LOCC only if it is a product basis~\cite{Horodecki}. Also, no basis of the subspace orthogonal to a state with orthogonal Schmidt number 3 or greater can be perfectly discriminated by LOCC \cite{DuanSep}. On the other hand, any three orthogonal maximally entangled states in $\bips{3}$ can be perfectly discriminated by LOCC~\cite{Nathanson}. In fact, if the number of dimensions is not restricted, one can find arbitrarily large sets of orthogonal maximally entangled states that can be perfectly discriminated by LOCC~\cite{Fan}. Contrary to intuition, states with more entanglement can sometimes be discriminated perfectly with LOCC while their less entangled counterparts cannot \cite{Horodecki}. Generally, however, a set of orthogonal multipartite states $S\subset\C^D$ can be perfectly discriminated with LOCC only if $\abs{S}\leq\frac{D}{d(S)}$, where $d(S)$ measures the average entanglement of the states in $S$ \cite{HMMOV}.

It is known that local projective measurements are sufficient to discriminate states from an orthonormal product basis with LOCC \cite{Rinaldis, ChenLi-ProdBasis}. Moreover, there is a polynomial-time (cubic in $\max \set{d_A, d_B}$) algorithm for deciding if states from a given orthonormal product basis of $\bip{d_A}{d_B}$ can be perfectly discriminated with LOCC \cite{Rinaldis}. The state discrimination problem for \emph{incomplete} orthonormal sets (i.e., orthonormal sets of states that do not span the entire space) seems to be harder to analyze. However, unextendible product bases might be an exception (although commonly referred to as ``bases'' these are in fact incomplete orthonormal sets). It is known that states from an unextendible product basis cannot be perfectly discriminated by finite LOCC \cite{BDMSST}. In fact, the same holds for any basis of a subspace spanned by an unextendible product basis in $\C^2 \x \C^2 \x \C^2$~\cite{DuanXinYing}. Curiously, there are only two families of unextendible product bases in $\bips{3}$, one of which is closely related to the domino states \cite{DMSST}.

The problem of state discrimination with \emph{asymptotic} LOCC has been studied less. It is known that states from an unextendible orthonormal product set cannot be perfectly discriminated with LOCC even asymptotically \cite{Rinaldis}. Reference \cite{KKB} gives a necessary condition for perfect asymptotic LOCC discrimination, and also shows that for perfectly discriminating states from an orthonormal product basis, asymptotic LOCC gives no advantage over finite LOCC. The latter result implies that the algorithm from \cite{Rinaldis} also covers the asymptotic case. On the other hand, even in some very basic instances of state discrimination it remains unclear whether asymptotic LOCC is superior to finite LOCC (see \cite{DuanSep,KKB} for specific sets of states). 

Another line of study originating from \cite{IBM} aims at understanding the difference between the classes of separable and LOCC operations. To this end, \cite{Cohen} constructs an $r$-round LOCC protocol implementing an arbitrary separable measurement whenever such a protocol exists. A different approach is to exhibit quantitative gaps between the two classes. To the best of our knowledge, only two quantitative gaps other than that of \cite{IBM} are known. References \cite{Koashi07, Koashi09} demonstrate a gap between the success probabilities achievable by bipartite separable and LOCC operations for unambiguously discriminating $\ket{00}$ from a fixed rank-2 mixed state. The largest known difference between the two classes is a gap of 0.125 between the achievable success probabilities for \emph{tripartite} EPR pair distillation \cite{ChitambarSep}. Moreover, as the number of parties grows, the gap approaches 0.37 \cite{ChitambarSep}.

At a first glance one might think that the nonlocality without entanglement phenomenon is related to quantum discord. However, the quantum discord value cannot be used to determine whether states from a given ensemble can be discriminated with LOCC \cite{Discord}.

Finally, if a set of orthogonal (product or entangled) states cannot be perfectly discriminated by LOCC, one can measure their nonlocality by considering how much entanglement is needed to achieve perfect discrimination \cite{Cohen-entanglement,BBKW}.

\section{Framework} \label{sec:Framework}

In this section we introduce a framework for proving lower bounds on the error probability of any LOCC measurement for discriminating bipartite states from a given set
\begin{equation}
  S \defeq \set{\ket{\psi_1}, \dotsc, \ket{\psi_n}} \subset \bip{d_A}{d_B}.
\end{equation}
We make no assumptions about the states $\ket{\psi_i}$.  In particular, they need not be product states or be mutually orthogonal.

From now on, $\mc{P}$ denotes an arbitrary LOCC protocol for discriminating states from $S$. In rough outline our argument proceeds as follows:
\begin{enumerate}
  \item We modify $\mc{P}$ so that it can be stopped when a specific amount of information $\ve$ has been obtained (see Section~\ref{sec:Interpolation}). This is done by terminating the protocol prematurely and possibly making the last measurement less informative (see Section~\ref{sec:Stopping}).
  \item When the information gain is $\ve$, we lower bound a measure of disturbance (defined in Section~\ref{sec:Delta}) by $\eta \ve$ for some constant $\eta$ (see Section~\ref{sec:Epsilon delta}).
  \item We show that at least two of the possible initial states have become nonorthogonal at this stage of the protocol, and we infer a lower bound on the error probability of $\mc{P}$ (see Section~\ref{sec:LowerBd}).
\end{enumerate}

Our framework reuses some ideas of the original approach~\cite{IBM}. However, instead of mutual information, we quantify how much an LOCC protocol has learned about the state using error probability. This allows us to replace the long mutual information analysis in the original paper with a simple application of Helstrom's bound. The idea of relating information gain and disturbance also comes from~\cite{IBM}. Here, we analyze this tradeoff using the nonlocality constant (see Definition~\ref{def:eta}) which can be applied to any set of states. In Section \ref{sec:Tilings} we give a method for lower bounding the nonlocality constant that applies specifically when $S$ is an \emph{orthonormal basis} of $\bip{d_A}{d_B}$. In Section \ref{sec:Dominoes} we apply this method for the domino states and some other related bases.

\subsection{Interpolated LOCC protocol} \label{sec:Interpolation}

Consider an arbitrary node in the tree representing the protocol $\mc{P}$. Let $m$ be the corresponding measurement record and let \mbox{$A \x B$} denote the Kraus operator that is applied to the initial state when this node is reached. Note that the output dimensions of operators $A$ and $B$ could be arbitrary.

The initial state $\ket{\psi_k}$ yields measurement record $m$ with probability
\begin{equation}
  p(m|\psi_k)
  \defeq \tr \bigl[ (A \x B)\ct (A \x B) \ket{\psi_k} \bra{\psi_k} \bigr]
   = \bra{\psi_k} (a \x b) \ket{\psi_k}
\end{equation}
where $a \defeq A\ct A \in \Pos(\C^{d_A})$ and $b \defeq B\ct B \in \Pos(\C^{d_B})$. Note that we need not concern ourselves with the arbitrary output dimensions of $A$ and $B$ from this point onward.  We use Bayes's rule and the uniformity of the probabilities $p(\psi_k)$ to obtain the probability that the initial state was $\ket{\psi_k}$ conditioned on the measurement record being $m$:
\begin{equation}
 p(\psi_k|m)
  = \frac{p(\psi_k) p(m|\psi_k)}{\sum_{j=1}^n p(\psi_j) p(m|\psi_j)}
  = \frac{             \bra{\psi_k} (a \x b) \ket{\psi_k}}
         {\sum_{j=1}^n \bra{\psi_j} (a \x b) \ket{\psi_j}} \,.
  \label{eq:p}
\end{equation}
At the root, the measurement record $m$ is the empty string and $p(\psi_k|m) = \frac{1}{n}$ for all~$k$. As we proceed toward the leaves, these probabilities fluctuate away from $\frac{1}{n}$.  For example, if $\mc{P}$ discriminates the states perfectly, the distribution reaches a Kronecker delta function.

For a given node $m$ let us define
\begin{equation}
  p_{\max}(m) \defeq \max_{k \in [n]} p(\psi_k|m).
\end{equation}
Let $\ve \defeq p_{\max}(m) - \frac{1}{n}$. Then $\ve$ characterizes the uniformity of the distribution $p(\psi_k|m)$ and thus the amount of information learned about the input state. The next theorem shows that we can modify the protocol $\mathcal{P}$ so that it can be stopped when some but not too much information has been learned.
While this idea originates from \cite{IBM}, we use a specific result from \cite{KKB}.

\begin{theorem}[Kleinmann, Kampermann, Bru\ss{} \cite{KKB}]
Let $\mc{P}$ be an LOCC protocol for discriminating states from a set $S$ of size $n$. For any $\ve > 0$ there exists an LOCC protocol $\mc{P}_{\ve}$ that has the same success probability as $\mc{P}$, but each branch of $\mc{P}_{\ve}$ has a node $m$ such that either
\begin{align}
  p_{\max}(m) &= \frac{1}{n} + \ve
  & \text{or} &&
  p_{\max}(m) &< \frac{1}{n} + \ve \text{ and $m$ is a leaf of $\mc{P}$}.
  \label{eq:Stage1}
\end{align}
\end{theorem}

\begin{proof}[Proof idea]
Let $u$ be a node in the protocol tree of $\mc{P}$ and let $v_1, \dotsc, v_m$ be the children of $u$. Assume that for some $i$ we have
\begin{equation}
  p_{\max}(u) < \frac{1}{n} + \ve < p_{\max}(v_i),
  \label{eq:Jump}
\end{equation}
which means that the measurement outcome corresponding to the edge $(u,v_i)$ is too informative. To rectify this, we break up the measurement at node $u$ into two steps.  We represent the outcomes of the first measurement by new nodes $\tilde{v}_1, \dotsc, \tilde{v}_m$ while the outcomes of the second measurement lead to the original nodes $v_1, \dotsc, v_m$ (see Figure~\ref{fig:Interpolation}).

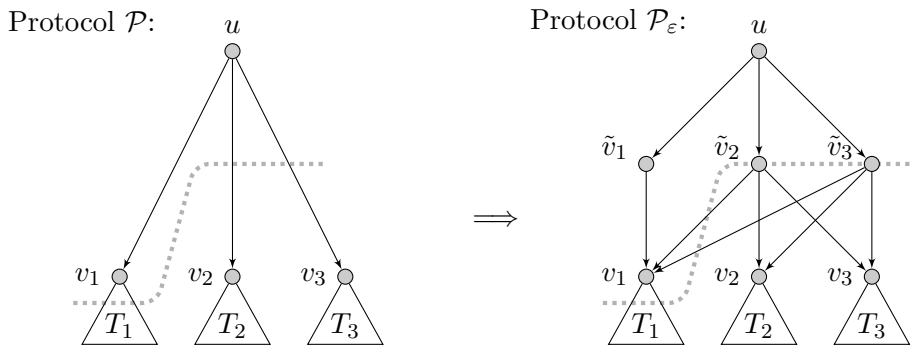
\begin{figure}[!ht]
\centering


\begin{tikzpicture}[
  circ/.style = {circle, draw = black, fill = gray!40, inner sep = 0mm, minimum size = 2mm},
  > = latex']

\def\w{1.5};
\def\h{1.5};

\def\tw{0.5};
\def\th{0.9};

\newcommand{\tree}[1]{
  \begin{scope}
    \def\d{0.5}
    \pgfsetcornersarced{\pgfpoint{2mm}{2mm}}
    \draw [dotted, ultra thick, gray!60] (\w+#1*\d,\h) -- (-\d,\h) -- (-\w+\d,-0.35) -- (-\w-\d-0.12,-0.35);
  \end{scope}

  \node (u) at (0,2*\h) [circ, label = $u$] {};

  \foreach \i in {1,2,3} {
    \pgfmathparse{\i-2};
    \let\x = \pgfmathresult;
    \node (v\i) at (\w*\x,0) [circ, label = left:$v_{\i}$] {};
    \draw (v\i) -- (\w*\x-\tw,-\th) -- (\w*\x+\tw,-\th) -- (v\i);
    \node at (\w*\x,-0.65) {$T_{\i}$};
  }
}

\begin{scope}[xshift = -3.5cm]
  \node at (-2,3.4) {Protocol $\mathcal{P}$:};
  \tree{-0.6}
  \foreach \i in {1,2,3}
    \draw [->] (u) -- (v\i);
\end{scope}

\node at (0,0.5*\h) {$\Longrightarrow$};

\newcommand{\shiftup}[1]{\underset{\phantom{\hat{\tilde{M}}}}{#1}}

\begin{scope}[xshift = 3.5cm]
  \node at (-2,3.4) {Protocol $\mathcal{P_{\varepsilon}}$:};
  \tree{1}
  \foreach \i in {1,2,3} {
    \pgfmathparse{\i-2};
    \let\x = \pgfmathresult;
    \node (w\i) at (\w*\x,\h) [circ, label = left:$\shiftup{\tilde{v}_{\i}}$] {};
    \draw [->] (u) -- (w\i);
  }

  \draw [->] (w1) -- (v1);
  \foreach \i in {2,3}
    \foreach \j in {1,2,3}
      \draw [->] (w\i) -- (v\j);
\end{scope}

\end{tikzpicture}

\caption{The protocol tree before (left) and after (right) splitting the measurement at node $u$ into two steps. (The graph on the right has been condensed for clarity, but it can be expanded into a tree by making a new copy of subtree $T_i$ for each incoming arc in $v_i$.) The amount of information learned in the first step is controlled by diluting the measurement operators, and the purpose of the second step is to complete the original measurement. The dotted line corresponds to the end of stage~I (see Definition \ref{def:stage1}).}
\label{fig:Interpolation}
\end{figure}

The first measurement interpolates between a completely uninformative trivial measurement and the original measurement at $u$. The interpolation parameters are chosen so that $p_{\max}(\tilde{v}_i) = \frac{1}{n} + \ve$ for all $i$ that satisfy Equation~(\ref{eq:Jump}). The second measurement depends on the outcome of the first measurement. It produces the same set of post-measurement states as the original measurement at $u$. Moreover, the total probability of obtaining each state is the same as in the case of the original measurement. After this we proceed according to the original protocol.

Protocol $\mc{P}_{\ve}$ is obtained from $\mc{P}$ by considering all branches of $\mc{P}$ and performing the above procedure at the closest node to the root that has a child satisfying Equation~(\ref{eq:Jump}). For more details see~\cite{KKB}.
\end{proof}

In the context of state discrimination, the possibility of interpolating a protocol to obtain some but not too much information is what distinguishes LOCC measurements from separable ones. In particular, a separable measurement for a set of states that cannot be distinguished by asymptotic LOCC cannot be divided into two steps, with the first yielding information precisely $\ve$ and the second completing the measurement (further details will be provided in a manuscript currently in preparation).

\subsection{Stopping condition} \label{sec:Stopping}

To control how much information the protocol has learned, we fix some $\ve > 0$ and stop the execution of $\mc{P}_{\ve}$ when we reach a node $m$ that satisfies the conditions in Equation~(\ref{eq:Stage1}).

\begin{definition}
We say that \emph{stage~I} of the protocol $\mc{P}_{\ve}$ is complete at the earliest point when Equation~(\ref{eq:Stage1}) is satisfied. \label{def:stage1}
\end{definition}

We choose $\ve < \frac{1}{n(n-1)}$ in our analysis.  Operationally, this means that none of the $n$ states has been eliminated at the end of stage~I, since 
\begin{equation}
  \min_{k \in [n]} p(\psi_k|m)
    \geq 1 - (n-1) p_{\max}(m)
    \geq \frac{1}{n} - (n-1) \ve > 0.
  \label{eq:pmin}
\end{equation}
This allows us to use Helstrom's bound to lower bound the probability of error (see Section~\ref{sec:LowerBd}).  It also ensures that the disturbance measure $\delta_{S}(a\x b)$ introduced in Section~\ref{sec:Delta} is well defined at $m$. All constraints imposed on the distribution $p(\psi_k|m)$ are summarized in Figure~\ref{fig:p}.

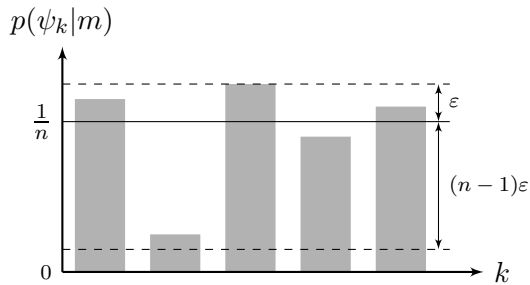
\begin{figure}[!ht]
\centering


\begin{tikzpicture}[> = latex', thick]

  \def\w{5}; 
  \def\h{3}; 
  \def\m{2}; 

  \def\e{0.5};
  \def\E{1.7};


  \draw[ycomb, color = gray!60, line width = 19] plot coordinates
  {(0.5, 2.3)
   (1.5, 0.5)
   (2.5, \m+\e)
   (3.5, 1.8)
   (4.5, 2.2)};


  \def\d{0.15} 

  \draw [        thin] (0,\m   ) to (\w+\d,\m);
  \draw [dashed, thin] (0,\m+\e) to (\w+\d,\m+\e);
  \draw [dashed, thin] (0,\m-\E) to (\w+\d,\m-\E);


  \path (0,\m) node [anchor = east] {$\frac{1}{n}$};
  \path (0, 0) node [anchor = east] {\scriptsize{$0$}};

  \def\x{\w};

  \draw [<->, thin] (\x,\m) -- (\x,\m+\e);
  \draw (\x,\m+0.5*\e) node [anchor = west] {\scriptsize{$\varepsilon$}};

  \draw [<->, thin] (\x,\m) -- (\x,\m-\E);
  \draw (\x,\m-0.5*\E) node [anchor = west] {\scriptsize{$(n-1)\varepsilon$}};


  \draw [->] (0, 0) -- (\w+4*\d,0) node [anchor = west ] {$k$};
  \draw [->] (0, 0) -- (0,\h) node [anchor = south] {$p(\psi_k|m)$};

\end{tikzpicture}

\caption{Probability distribution $p(\psi_k|m)$ at the end of stage~I. For all $k$ we have $\frac{1}{n} + \ve \geq p(\psi_k|m) \geq \frac{1}{n} - (n-1) \ve > 0$ where the first inequality is tight for some $k$.}
\label{fig:p}
\end{figure}

Since the error probability of the protocol $\mc{P}_{\ve}$ is a weighted average of error probabilities of individual branches, it suffices to lower bound these individual error probabilities.  For any branch that terminates without a node satisfying 
\begin{equation}
  p_{\max}(m) = \frac{1}{n} + \ve \,, 
  \label{eq:pmax}
\end{equation}
we can put a large lower bound on the error probability.  In particular, for the optimal choice $\ve = \frac{2}{3} \frac{1}{n(n-1)}$ of Theorem~\ref{thm:eta} with $n \geq 2$,
\begin{equation}
  p_{\error}(m)
  \geq 1 - p_{\max}(m)
  > 1 - \biggl( \frac{1}{n} + \ve \biggr)
  = 1 - \frac{1}{n} - \frac{2}{3} \frac{1}{n(n-1)}
  \geq \frac{1}{6}, 
\end{equation}
which is much higher than the lower bound we obtain for other branches.
We now consider the remaining case where stage~I ends with a node 
satisfying Equation~(\ref{eq:pmax}).

\subsection{Measure of disturbance} \label{sec:Delta}

Now we show that at least two possible post-measurement states $(A \x B)\ket{\psi_i}$ and $(A \x B) \ket{\psi_j}$ are nonorthogonal at the end of stage~I, and lower bound their overlap quantitatively. Assuming that the initial state was $\ket{\psi_i} \in S$, the normalized post-measurement state at the node with corresponding measurement operator $A \x B$ is
\begin{equation}
  \ket{\phi_i} \defeq \frac{\bigl( A \x B \bigr)\ket{\psi_i}}
                       {\sqrt{\bra{\psi_i} (a \x b) \ket{\psi_i}}}
\end{equation}
where $a \defeq A\ct A$ and $b \defeq B\ct B$. Note that  
$\bra{\psi_i} (a \x b) \ket{\psi_i} > 0$ for all $i \in [n]$ because, 
from Equations~(\ref{eq:pmin}) and~(\ref{eq:p}), 
$0 < \min_{k \in [n]} p(\psi_k|m)
   = \min_{k \in [n]}
     \frac{             \bra{\psi_k} (a \x b) \ket{\psi_k}}
          {\sum_{j=1}^n \bra{\psi_j} (a \x b) \ket{\psi_j}}$.

\begin{definition}
The \emph{disturbance} caused by the operator $a \x b$ on the set of states $S$ is defined as 
\label{def:delta}
\begin{equation}
  \delta_S(a \x b)
  \defeq \max_{i \neq j} \abs{\braket{\phi_i}{\phi_j}}
   = \max_{i \neq j}
     \frac{\abs {\bra{\psi_i} (a \x b) \ket{\psi_j}}}
          {\sqrt{\bra{\psi_i} (a \x b) \ket{\psi_i}
                 \bra{\psi_j} (a \x b) \ket{\psi_j}}}.
  \label{eq:delta}
\end{equation}
\end{definition}

Note that $\delta_S(a \x b)$ measures the nonorthogonality of the post-measurement states. If the initial states $\ket{\psi_i}$ were orthogonal then $\delta_S(a \x b)$ indeed characterizes the disturbance caused by $a \x b$.

Since $\braket{\phi_i}{\phi_j}$ can be expressed in terms of the operators $a = A^\dag A$ and $b = B^\dag B$, from now on we no longer explicitly use the measurement operators $A$ and $B$.

\subsection{Disturbance/information gain trade-off} \label{sec:Epsilon delta}

Now we define the {\em nonlocality constant} and show that it relates $\delta$ (the disturbance caused at the end of stage~I) to $\ve$ (the amount of information learned).

\begin{definition}
The \emph{nonlocality constant} of $S$ is the supremum over all $\eta$ such that
for all $a \in \Pos(\C^{d_A}), b \in \Pos(\C^{d_B})$ and for all $i$ satisfying $\bra{\psi_i} (a \x b) \ket{\psi_i} \neq 0$, 
\begin{equation}
  \eta \cdot \biggl(
    \frac{      \max_{k \in [n]} \bra{\psi_k} (a \x b) \ket{\psi_k}}
         {\;\;\;\sum_{j \in [n]} \bra{\psi_j} (a \x b) \ket{\psi_j}} - \frac{1}{n}
  \biggr)
  \leq \delta_S(a \x b) \,.
  \label{eq:eta}
\end{equation}
\end{definition}

\noindent Equivalently, if $G_{ij} \defeq \bra{\psi_i} (a \x b) \ket{\psi_j}$ for $i,j \in [n]$ then 
\begin{equation}
  \eta \defeq \inf_{a, b}
\left\{
\frac{
\max_{i \neq j} \dfrac{\abs{G_{ij}}}{\sqrt{G_{ii} G_{jj}}}
}{
\dfrac{\max_k G_{kk}}{\sum_{j=1}^n G_{jj}} - \dfrac{1}{n}
}
\right\}
\end{equation}
where the infimum is over all $a \in \Pos(\C^{d_A})$ and $b \in \Pos(\C^{d_B})$ such that $G_{ii} \neq 0$ for all $i \in [n]$.
\label{def:eta}

Recall from Section~\ref{sec:Stopping} that we stop the LOCC protocol at the end of stage~I in a node $m$ where the condition in Equation~(\ref{eq:pmax}) is satisfied for some $\ve \in \bigl( 0, \frac{1}{n(n-1)} \bigr)$. Let $a \x b$ be the operator corresponding to node $m$ and let $\delta \defeq \delta_S(a \x b)$ be the disturbance caused.

\begin{lemma}[Disturbance/information gain trade-off]
\label{lem:Epsilon delta}
The amount of information $\ve$ learned at the end of stage I lower bounds the disturbance $\delta$ as
\begin{equation}
  \eta \, \ve \leq \delta
  \label{eq:Epsilon delta}
\end{equation}
where $\eta$ is the nonlocality constant of $S$ (see Definition~\ref{def:eta}).
\end{lemma}

\begin{proof}
This immediately follows from the definitions of $\ve$ and $\eta$:
\begin{align}
  \eta \, \ve
  = \eta \biggl( \max_{k \in [n]} p(\psi_k|m) - \frac{1}{n} \biggr)
  = \eta \biggl(
           \frac{\max_{k \in [n]} \bra{\psi_k} (a \x b) \ket{\psi_k}}
                {\sum_{j=1}^n \bra{\psi_j} (a \x b) \ket{\psi_j}} - \frac{1}{n}
         \biggr)
  \leq \delta
\end{align}
where we have used Equations~(\ref{eq:pmax}), (\ref{eq:p}), and~(\ref{eq:eta}).
\end{proof}

\subsection{Lower bounding the error probability} \label{sec:LowerBd}

In this section we use Lemma~\ref{lem:Epsilon delta} to lower bound the error probability of any LOCC measurement for discriminating states from the set $S$.

Note that Equation~(\ref{eq:Epsilon delta}) together with the definition of $\delta$ implies that at the end of stage~I there are two distinct post-measurement states $\ket{\phi_i}$ and $\ket{\phi_j}$ such that 
\begin{equation}
  \abs{\braket{\phi_i}{\phi_j}} = \delta \geq \eta \, \ve.
\end{equation}
As discussed in Section~\ref{sec:Stopping}, our choice of $\ve$ guarantees that $p(\psi_i|m)$ and $p(\psi_j|m)$ are both strictly positive. Thus we can use the following result to lower bound the error probability:

\begin{fact*}[Helstrom bound~{\cite[pp.113]{Helstrom}}]
Suppose we are given state $\ket{\Phi_0}$ with probability $q_0$ and state $\ket{\Phi_1}$ with probability $q_1 = 1 - q_0$. Any measurement trying to discriminate the two cases errs with probability at least
\begin{equation}
  Q(q_0, q_1, \delta)
  \defeq \frac{1}{2} \bigl( 1 - \sqrt{1 - 4 q_0 q_1 \delta^2} \bigr)
  \geq q_0 q_1 \delta^2, \label{eq:Q}
\end{equation}
where $\delta = \abs{\braket{\Phi_0}{\Phi_1}}$ is the overlap between the two states, and the inequality follows from $1-\sqrt{1-x^2} \geq \frac{1}{2} x^2$ for $x \in [0,1]$.
\end{fact*}

As $\ve$ increases, the disturbance (thus the overlap between some $\ket{\phi_i}$ and $\ket{\phi_j}$) increases, but the lower bound on the probabilities $p(\psi_i|m)$ and $p(\psi_j|m)$ decreases.  The choice $\ve = \frac{2}{3}\frac{1}{n(n-1)}$ gives a lower bound on the error probability as follows.

\begin{theorem}
\label{thm:eta}
Let $S$ be a set of quantum states in $\bip{d_A}{d_B}$ of size $n \geq 2$. Any LOCC measurement for discriminating states drawn uniformly from $S$ errs with probability
\begin{equation}
  p_{\error} \geq \frac{2}{27} \, \frac{\eta^2}{n^5}
  \label{eq:perror}
\end{equation}
where $\eta$ is the nonlocality constant of $S$ (see Definition~\ref{def:eta}).
\end{theorem}

\begin{proof}
At the end of stage~I there are two post-measurement states $\ket{\Phi_0}$ and $\ket{\Phi_1}$ with overlap $\delta$. Let $p_0$ and $p_1$ be the posterior probabilities of these states. To lower bound the error probability of $\mc{P}_{\ve}$ (thus that of $\mc{P}$), we give Alice and Bob extra power at this point:
\begin{itemize}
  \item if the actual input state does not lead to $\ket{\Phi_0}$ or $\ket{\Phi_1}$, we assume that Alice and Bob succeed with certainty;
  \item otherwise Alice and Bob are allowed to perform the best joint measurement to discriminate the states $\ket{\Phi_0}$ and $\ket{\Phi_1}$.
\end{itemize}
For fixed $\ve$ and probabilities $p_0$ and $p_1$, we can lower bound the error probability by the following expression:
\begin{equation}
  P(p_0, p_1, \ve)
  \defeq (p_0 + p_1) \cdot Q
     \Bigl(
       \tfrac{p_0}{p_0 + p_1},
       \tfrac{p_1}{p_0 + p_1},
       \delta
     \Bigr).
\end{equation}
Using Equation~(\ref{eq:Q}) and the inequality $\delta \geq \eta \, \ve$ from Lemma~\ref{lem:Epsilon delta}, we get that
\begin{equation}
  P(p_0, p_1, \ve)
  \geq \frac{p_0 p_1}{p_0 + p_1} (\eta \, \ve)^2.
  \label{eq:p0 and p1}
\end{equation}

Recall that we stop the protocol at a point where we are guaranteed that $0 < \ve < \frac{1}{n(n-1)}$ and, by Equations (\ref{eq:pmin}) and (\ref{eq:pmax}),
\begin{equation}
  \frac{1}{n} - (n-1)\ve \leq p_i \leq \frac{1}{n} + \ve
\end{equation}
for all $i$. Given these constraints on $p_0$ and $p_1$, we can choose the $\ve$ that maximizes $P(p_0, p_1, \ve)$ and guarantee that the error probability in the branch of the LOCC protocol being considered satisfies
\begin{equation}
  p_{\error} \geq 
  \max_{\ve \in \bigl( 0, \frac{1}{n(n-1)} \bigr)} \;\;
  \min_{p_0, p_1 \in \bigl[ \frac{1}{n} - (n-1)\ve, \frac{1}{n} + \ve \bigr]} \;\;
  P(p_0, p_1, \ve).
\end{equation}
From Equation~(\ref{eq:p0 and p1}) we get
\begin{equation}
  p_{\error}  \geq 
  \max_{\ve \in \bigl( 0, \frac{1}{n(n-1)} \bigr)} \;\;
  \min_{p_0, p_1 \in \bigl[ \frac{1}{n} - (n-1)\ve, \frac{1}{n} + \ve \bigr]} \;\;
  \frac{p_0 p_1}{p_0 + p_1} (\eta \, \ve)^2.
\end{equation}
The minimum is attained when $p_0 = p_1 = \frac{1}{n} - (n-1) \ve$ (i.e., the probabilities are equal and as small as possible), so the problem simplifies to
\begin{equation}
  p_{\error} \geq
  \max_{\ve \in \bigl( 0, \frac{1}{n(n-1)} \bigr)}
    \frac{1}{2} \biggl( \frac{1}{n} - (n-1)\ve \biggr)
    (\eta \, \ve)^2 
  \geq \frac{2}{27} \frac{\eta^2}{n^3(n-1)^2}
  \geq \frac{2}{27} \frac{\eta^2}{n^5}
\end{equation}
where the value
\begin{equation}
  \ve = \frac{2}{3} \frac{1}{n(n-1)}
  \label{eq:epsilon}
\end{equation}
achieves the maximum.
\end{proof}

Theorem \ref{thm:eta} shows that any LOCC protocol for discriminating states from $S$ errs with probability proportional to $\eta^2$, justifying the name ``nonlocality constant.''

\section{Bounding the nonlocality constant} \label{sec:Tilings}

The framework described in Section~\ref{sec:Framework} reduces the problem of bounding the error probability for discriminating bipartite states by LOCC to the one of bounding the nonlocality constant $\eta$ (see Theorem~\ref{thm:eta}). This reduction holds for any set of pure states $S$. In this section we assume that $S$ is an \emph{orthonormal basis} of $\bip{d_A}{d_B}$ and provide tools for bounding the nonlocality constant. In particular, we bound $\eta$ in terms of another quantity that we call ``rigidity''.

For the remainder of the paper we represent pure states from $\bip{d_A}{d_B}$ using ``tiles'' in a $d_A \times d_B$ grid. We first introduce some notations related to tilings in Section~\ref{sec:Definitions}. Then we define rigidity and relate it to the nonlocality constant $\eta$ in Section~\ref{sec:Box}.  Section~\ref{sec:Pair of tiles} provides a tool, the ``pair of tiles'' lemma, that we use to bound rigidity for specific sets of states in Section~\ref{sec:Dominoes}.

\subsection{Definitions} \label{sec:Definitions}

Given a fixed orthonormal basis $\set{\ket{i} \colon i \in [d]}$, define the \emph{support} of a pure state $\ket{\psi} \in \C^d$ as 
\begin{equation}
  \supp \ket{\psi} \defeq \set{i \in [d] \colon \braket{i}{\psi} \neq 0}.
\end{equation}
If $\ket{\psi} \in \bip{d_A}{d_B}$ then $\supp \ket{\psi} \subseteq [d_A] \times [d_B]$. Consider $[d_A] \times [d_B]$ as a rectangular grid of size $d_A \times d_B$. Any region that corresponds to a submatrix of this grid is called a tile. More formally, a \emph{tile} is a subset $T \subseteq [d_A] \times [d_B]$ such that $T = R \times C$ for some $R \subseteq [d_A]$ and $C \subseteq [d_B]$. (Note that a tile is not necessarily a contiguous region of the grid.) We use $\rows(T) = R$ and $\cols(T) = C$ to denote the \emph{rows} and \emph{columns} of this tile, respectively, and we use $\abs{T}$ to denote the \emph{size} or the \emph{area} of $T$. If $\ket{\psi} = \ket{\alpha} \ket{\beta}$ is a product state, then $\supp \ket{\psi} = \supp \ket{\alpha} \times \supp \ket{\beta}$ and thus $\supp \ket{\psi}$ is a tile, which we call the \emph{tile induced by} $\ket{\psi}$.

We say that an orthonormal set of product states $S \subset \bip{d_A}{d_B}$ \emph{induces a tiling} of a $d_A \times d_B$ grid if the tiles induced by the states in $S$ are either disjoint or identical. Note that if $S$ is an orthonormal basis of $\bip{d_A}{d_B}$, then a tile of area $L$ is induced by $L$ states that form a \emph{basis} of that tile. In a \emph{domino-type} tiling, every tile has area $1$ or $2$.

For a given tiling $T$ of a $d_A \times d_B$ grid let us define the corresponding \emph{row graph} as follows: its vertex set is $[d_A]$ with two vertices $i$ and $j$ adjacent if and only if there exists a column $c$ such that $(i,c)$ and $(j,c)$ belong to the same tile. The \emph{column graph} of a tiling is defined similarly. We say that a tiling is \emph{irreducible} if its row graph and its column graph are both connected. The \emph{diameter} of the tiling $T$ is the maximum of the diameters of its row and column graphs.  See Figure \ref{fig:Domino-type} for an example.

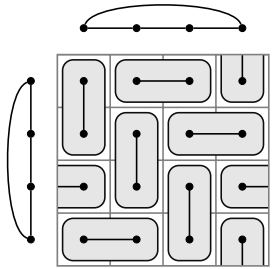
\begin{figure}[!ht]
\centering


\def\step{20pt} 

\begin{tikzpicture}[
  domino/.style = {rectangle, rounded corners = 0.2*\step, draw = black!95, fill = black!10},
  circ/.style = {circle, draw = black, fill = black, inner sep = 0mm, minimum size = 0.13*\step},
  gridlines/.style = {gray, semithick}
]

  \newcommand{\domino}[4]{
    \node[semithick, domino, minimum height = #3*\step + 0.8*\step, minimum width = #4*\step + 0.8*\step]
          at (#1*\step, #2*\step) {};
    \draw[semithick]
         (#1*\step - #4*\step/2, #2*\step - #3*\step/2) --
         (#1*\step + #4*\step/2, #2*\step + #3*\step/2);
  }

  \newcommand{\hdomino}[2]{
    \pgfmathparse{#1+1.0};
    \let\x = \pgfmathresult;
    \pgfmathparse{#2+0.5};
    \let\y = \pgfmathresult;
    \domino{\x}{\y}{0}{1}
  }

  \newcommand{\vdomino}[2]{
    \pgfmathparse{#1+0.5};
    \let\x = \pgfmathresult;
    \pgfmathparse{#2+1.0};
    \let\y = \pgfmathresult;
    \domino{\x}{\y}{1}{0}
  }

  \begin{scope}[fill = white]

    \fill[clip] (0,0) rectangle (4*\step,4*\step);

    \draw[step = \step, gridlines] (0,0) grid (4*\step, 4*\step);

    \hdomino{0}{0}
    \hdomino{2}{2}
    \hdomino{1}{3}

    \hdomino{3}{1}
    \hdomino{-1}{1}

    \vdomino{2}{0}
    \vdomino{1}{1}
    \vdomino{0}{2}

    \vdomino{3}{-1}
    \vdomino{3}{3}

  \end{scope}

  \draw[gridlines] (0,0) -- (4*\step, 0) -- (4*\step,4*\step) -- (0, 4*\step) -- cycle;

  \foreach \i in {0,1,2,3}{
    \node (r\i) at (-\step/2,\i*\step+\step/2) [circ] {};
    \node (c\i) at (\i*\step+\step/2,4*\step+\step/2) [circ] {};
    \foreach \j in {0,1,2,3}{
      \node at (\i*\step+\step/2,\j*\step+\step/2) [circ] {};
    }
  };

  \def\r{0.6*\step} 

  \draw[semithick] (r0) -- (r1) -- (r2) -- (r3) .. controls +(200:\r) and +(160:\r) .. (r0);
  \draw[semithick] (c0) -- (c1) -- (c2) -- (c3) .. controls +(110:\r) and +( 70:\r) .. (c0);
\end{tikzpicture}

\caption{A domino-type tiling and the corresponding row and column graphs. This tiling is irreducible and has diameter two.}
\label{fig:Domino-type}
\end{figure}

Without loss of generality we consider only irreducible tilings.  Reducible tilings can be broken down into several smaller components without disturbing the underlying states. To do this, both parties simply perform a projective measurement with respect to the subspaces corresponding to the different components of the row and column graphs.

Note that in general, a tiling is not invariant under local unitaries.  In particular, the irreducibility of the tiling induced by a given set of states is a basis-dependent property. The most extreme example of this phenomenon is the case of the standard basis. It induces a completely reducible tiling that consists only of $1 \times 1$ tiles. However, if both parties apply a generic local unitary transformation, the resulting tiling consists only of a single tile of maximal size.

\subsection{Lower bounding the nonlocality constant using rigidity} \label{sec:Box}

In this section we assume that $S$ is an orthonormal basis of $\bip{d_A}{d_B}$ (so in particular, $n = d_A d_B$) and discuss a particular strategy for lower bounding $\eta$ for such $S$. We apply this strategy to several sets of orthonormal product bases in Section~\ref{sec:Dominoes}.

We bound $\eta$ (quantifying a disturbance/strength tradeoff) by considering a quantitative property of the set $S$ called {\em rigidity}.  Intuitively, we call a measurement operator strong if it is far from being proportional to the identity matrix; a set of states $S$ is rigid if there exists a strong measurement that leaves the set undisturbed.  We formalize this as follows (recall that $\norm{\cdot}_{\max}$ denotes the largest entry of a matrix in absolute value):

\begin{definition}
For an orthonormal basis $S$, if there is a constant $c$ such that for all $a \in \Pos(\C^{d_A})$, $b \in \Pos(\C^{d_B})$ and for all $i$ such that $\bra{\psi_i} (a \x b) \ket{\psi_i} \neq 0$, 
\begin{equation}
  \norm{\frac{a \x b}{\tr(a \x b)} - \frac{\id}{n}}_{\max} \!\! \leq \, c \cdot \delta_S(a \x b),
  \label{eq:Rigidity}
\end{equation}
we say $S$ is \emph{$c$-rigid}, or $c$ is an upper bound on the rigidity of $S$.  
\label{def:Rigidity}
\end{definition}

When $S$ is rigid, the states can remain unchanged despite application of a strong measurement.  For example, a tensor product basis is not $c$-rigid for any finite $c$ (i.e., such a basis is arbitrarily rigid).  In contrast, if $c$ is small, then \emph{any} strong measurement disturbs the set $S$, and Equation~(\ref{eq:Rigidity}) quantifies how weak a measurement operator $a \x b$ must be for the disturbance $\delta_S(a \x b)$ to be small.  

We now relate upper bounds on the rigidity of $S$ to lower bounds on its nonlocality constant:

\begin{lemma}\label{lem:ceta}
Let $S$ be an orthonormal basis of $\bip{d_A}{d_B}$. If $S$ is $c$-rigid then 
\begin{equation}
  \eta \geq \frac{1}{cL}.
  \label{eq:eta bound}
\end{equation}
where $L$ is the size of the largest tile corresponding to states in $S$.
\end{lemma}

\begin{proof}
If $S$ is $c$-rigid, then for any $a \in \Pos(\C^{d_A})$ and $b \in \Pos(\C^{d_B})$ (such that $\bra{\psi_k} (a \x b) \ket{\psi_k} \neq 0$ for all $k \in [n]$), we have
\begin{equation}
  \frac{a \x b}{\tr(a \x b)} - \frac{\id}{n} = c M \cdot \delta_S(a \x b)
\end{equation}
for some Hermitian matrix $M \in \Lin(\bip{d_A}{d_B})$ with $\norm{M}_{\max} \leq 1$. From this we get
\begin{align}
  \max_{k \in [n]} \bra{\psi_k}  \frac{a \x b}{\tr (a \x b)} \ket{\psi_k}
  - \frac{1}{n}
  &= c \max_{k \in [n]} \bra{\psi_k} M \ket{\psi_k} \cdot \delta_S(a \x b) \\
  &\leq cL \cdot \delta_S(a \x b).
\end{align}
By the definition of $\eta$ (Equation~(\ref{eq:eta})) and the fact that $\tr (a \x b) = \sum_{j \in [n]} \bra{\psi_j} (a \x b) \ket{\psi_j}$ for any orthonormal basis $S$, we get the desired inequality.
\end{proof}

Putting Lemma~\ref{lem:ceta} and Theorem~\ref{thm:eta} together gives the following:
 
\begin{theorem}
\label{thm:cL}
Let $S$ be an orthonormal basis of $\bip{d_A}{d_B}$. If $S$ is $c$-rigid then any LOCC measurement for discriminating states from $S$ errs with probability
\begin{equation}
  p_{\error} \geq \frac{2}{27} \, \frac{1}{(cL)^2 n^5}
\end{equation}
where $L$ is the size of the largest tile of $S$.
\end{theorem}

\subsection{The ``pair of tiles'' lemma} \label{sec:Pair of tiles}

In this section we present a lemma that serves as our main tool for bounding rigidity.

\begin{lemma}
\label{claim:UV}
Let $U \in \U(m), V \in \U(n)$, and define $\ket{\varphi_i} \defeq U \ket{i}$ for $i \in [m]$ and $\ket{\psi_j} \defeq V \ket{j}$ for $j \in [n]$. Then for any $M \in \Lin(\C^n, \C^m)$ we have
\begin{equation}
  \sqrt{mn} \cdot
  \max_{i,j}
  \abs{\bra{\varphi_i} M \ket{\psi_j}} \geq
  \max_{k,l}
  \abs{M_{kl}}.
\end{equation}
\end{lemma}

The main idea of the proof is that a unitary change of basis can only increase the largest entry of a vector by a multiplicative factor depending on the dimension of the vector.

\begin{proof}
Let us define a mapping $\vec \colon \Lin(\C^n, \C^m) \to \bip{n}{m}$ as 
\begin{equation}
  \vec \colon \ket{i} \bra{j} \mapsto \ket{i} \ket{j} 
\end{equation}
for $i \in [m]$ and $j \in [n]$ and extend it by linearity over $\C$. One can check that $\vec(A X B) = (A \x B\tp) \vec(X)$. Using this and basic inequalities between the $2$-norm and the $\infty$-norm, we get
\begin{align}
  \max_{i,j} \abs{\bra{\varphi_i} M \ket{\psi_j}}
  &=    \biggl\| \vec \Bigl( \sum_{i,j} \bra{\varphi_i} M \ket{\psi_j} \ket{i} \bra{j} \Bigr) \biggr\|_{\infty} \\
  &=    \biggl\| \vec \Bigl( \sum_{i,j} \bra{i} U\ct M V \ket{j} \ket{i} \bra{j} \Bigr) \biggr\|_{\infty} \\
  &=    \bigl\| \vec (U\ct M V) \bigr\|_{\infty} \\
  &=    \bigl\| (U\ct \x V\tp) \vec(M) \bigr\|_{\infty} \\
  &\geq \frac{1}{\sqrt{mn}} \bigl\| (U\ct \x V\tp) \vec(M) \bigr\|_2 \\
  &=    \frac{1}{\sqrt{mn}} \bigl\| \vec(M) \bigr\|_2 \\
  &\geq \frac{1}{\sqrt{mn}} \bigl\| \vec(M) \bigr\|_{\infty} \\
  &=    \frac{1}{\sqrt{mn}} \max_{k,l} \abs{M_{kl}},
\end{align}
as desired.
\end{proof}

Let us restate Lemma~\ref{claim:UV} using the language of tilings:

\begin{lemma}\label{claim:Regions}
Let $R_1, R_2 \subseteq [d_A] \times [d_B]$ be two arbitrary regions of a $d_A \times d_B$ grid, and $\set{\ket{\varphi_i}}_{i=1}^{\abs{R_1}}$ and $\set{\ket{\psi_j}}_{j=1}^{\abs{R_2}} \subset \bip{d_A}{d_B}$ be their bases (here $\ket{\varphi_i}$ and $\ket{\psi_j}$ need not be product states). Then for any matrices $a \in \Lin(C^{d_A})$ and $b \in \Lin(C^{d_B})$ we have
\begin{equation}
  \sqrt{\abs{R_1} \cdot \abs{R_2}} \;
  \max_{i,j} \abs{\bra{\varphi_i} (a \x b) \ket{\psi_j}}
  \; \geq \max_{\substack{(r_1,c_1) \in R_1\\(r_2,c_2)\in R_2}}
  \abs{a_{r_1 r_2}} \cdot \abs{b_{c_1 c_2}}.
\end{equation}
\end{lemma}

This follows from Lemma~\ref{claim:UV} by restricting $\ket{\varphi_i}$ and $\ket{\psi_j}$ to regions $R_1$ and $R_2$, respectively, and choosing $M$ to be a submatrix of $a \x b$ with rows determined by $R_1$ and columns by $R_2$.

\begin{proof}
For $t \in \set{1,2}$ let us enumerate the cells of region $R_t$ by integers from $\set{1, \dotsc, \abs{R_t}}$ arbitrarily, and let $(r_t(i), c_t(i))$ be the coordinates of the $i$th cell of region $R_t$. Let
\begin{equation}
  \Pi_t \defeq \sum_{i=1}^{\abs{R_t}} \ketbra{i}{r_t(i), c_t(i)}
\end{equation}
be a linear operator that restricts the space $\bip{d_A}{d_B}$ to region $R_t$. Then $\ket{\vp'_i} \defeq \Pi_1 \ket{\vp_i}$ is the restriction of $\ket{\vp_i}$ to region $R_1$ and $\ket{\psi'_i} \defeq \Pi_2 \ket{\psi_i}$ is the restriction of $\ket{\psi_i}$ to $R_2$. Also, let $M \defeq \Pi_1 (a\x b) \Pi_2\ct$. 

Note that for all $i \in \set{1, \dotsc, \abs{R_1}}$ we have $\Pi_1\ct \Pi_1 \ket{\vp_i} = \ket{\vp_i}$ since the support of $\ket{\vp_i}$ lies entirely within region $R_1$ and $\Pi_1\ct\Pi_1$ is the projection onto $R_1$. Similarly, $\Pi_2\ct\Pi_2\ket{\psi_j}=\ket{\psi_j}$ for all $j\in\set{1,\dotsc,\abs{R_2}}$. Hence
\begin{equation}
  \bra{\vp_i}(a \x b)\ket{\psi_j} = 
  \bra{\vp_i}\Pi\ct_1\Pi_1(a \x b)\Pi_2 \ct\Pi_2\ket{\psi_j} = 
  \bra{\vp'_i}M\ket{\psi'_j}
\end{equation}
for all $i$ and $j$. Finally, we apply Lemma~\ref{claim:UV} to $\{\ket{\vp'_i}\}_{i=1}^{\abs{R_1}}$, $\{\ket{\psi'_j}\}_{j=1}^{\abs{R_2}}$, and $M$:
\begin{align*}
  \sqrt{\abs{R_1} \cdot \abs{R_2}} \;
    \max_{i,j} \abs{\bra{\varphi_i} (a \x b) \ket{\psi_j}} 
  &=\sqrt{\abs{R_1} \cdot \abs{R_2}} \;
    \max_{i,j} \abs{\bra{\vp'_i}M\ket{\psi'_j}}\\
  &\geq \max_{k,l} \abs{M_{kl}} \\ 
  &= \max_{k,l} \abs{\bra{k}\Pi_1 (a\x b) \Pi_2\ct\ket{l}}\\
  &= \max_{k,l} \big|\bra{r_1(k)} \, a \, \ket {r_2(l)} \big| \cdot
                \big|\bra{c_1(k)} \, b \, \ket {c_2(l)} \big|\\
  &= \max_{\substack{(r_1,c_1)\in R_1\\(r_2,c_2)\in R_2}}
  \abs{a_{r_1 r_2}} \cdot \abs{b_{c_1 c_2}}
\end{align*}
and the result follows.
\end{proof}

When regions $R_1$ and $R_2$ are two distinct tiles from the tiling induced by $S$, we can use Lemma~\ref{claim:Regions} to get the following result:

\begin{lemma}[``Pair of tiles" Lemma]
\label{lem:PairOfTiles}
Let $T_1$ and $T_2$ be two distinct tiles in the tiling induced by $S$, and let
$a \in \Pos(\C^{d_A})$ and $b \in \Pos(\C^{d_B})$. Then
\begin{equation}
  \sqrt{\abs{T_1} \cdot \abs{T_2}} \;
  \delta_S(a \x b) \tr(a \x b)
  \geq \abs{a_{r_1 r_2}} \cdot \abs{b_{c_1 c_2}}
  \label{eq:Tiles}
\end{equation}
for any $r_t \in \rows(T_t)$ and $c_t \in \cols(T_t)$ where $t \in \set{1,2}$.
\end{lemma}

\begin{proof}
We relax the inequality in Lemma~\ref{claim:Regions} by observing that
\begin{align}
  \delta_S(a \x b)
  &\geq \frac{\max_{i \neq j} \abs{\bra{\psi_i} (a \x b) \ket{\psi_j}}}{\norm{a \x b}_{\infty}}
   \geq \frac{\max_{i \neq j} \abs{\bra{\psi_i} (a \x b) \ket{\psi_j}}}{\tr (a \x b)}
\end{align}
which easily follows from the definition of $\delta_S(a \x b)$ in Equation~(\ref{eq:delta}).
\end{proof}

Note that the tiles $T_1$ and $T_2$ in Lemma~\ref{lem:PairOfTiles} have to be distinct since the maximization in the definition of $\delta_S(a \x b)$ is performed only over pairs of distinct states. This lemma will be used later to bound the off-diagonal entries of $a \x b$ (see Figure~\ref{fig:Pair}).

\begin{figure}[!ht]
\centering


\def\step{20pt} 

\begin{tikzpicture}[
  domino/.style = {rectangle, rounded corners = 0.2*\step, draw = black!95, fill = black!10},
  circ/.style = {circle, draw = black, fill = black, inner sep = 0mm, minimum size = 0.13*\step},
  gridlines/.style = {gray, semithick}
]

  \newcommand{\domino}[4]{
    \node[semithick, domino, minimum height = #3*\step + 0.8*\step, minimum width = #4*\step + 0.8*\step]
          at (#1*\step, #2*\step) {};
  }

  \newcommand{\hdomino}[2]{
    \pgfmathparse{#1+0.5};
    \let\x = \pgfmathresult;
    \pgfmathparse{#2};
    \let\y = \pgfmathresult;
    \domino{\x}{\y}{0}{1}
  }

  \newcommand{\vdomino}[2]{
    \pgfmathparse{#1};
    \let\x = \pgfmathresult;
    \pgfmathparse{#2+0.5};
    \let\y = \pgfmathresult;
    \domino{\x}{\y}{1}{0}
  }

  \def\dx{0.7}

  \newcommand{\Hline}[4]{
    \pgfmathparse{#3};
    \let\y = \pgfmathresult;
    \pgfmathparse{#1-\dx};
    \let\xa = \pgfmathresult;
    \pgfmathparse{#2+\dx};
    \let\xb = \pgfmathresult;
    \draw (\xa*\step, \y*\step) -- (\xb*\step, \y*\step);
    \node at (\xa*\step - 0.4*\step, \y*\step) {#4};
  }

  \newcommand{\Vline}[4]{
    \pgfmathparse{#3};
    \let\x = \pgfmathresult;
    \pgfmathparse{#1+\dx};
    \let\ya = \pgfmathresult;
    \pgfmathparse{#2-\dx};
    \let\yb = \pgfmathresult;
    \draw (\x*\step, \ya*\step) -- (\x*\step, \yb*\step);
    \node at (\x*\step, \ya*\step + 0.3*\step) {#4};
  }


  \hdomino{2}{2}
  \vdomino{0}{2}

  \def\ra{3}
  \def\rb{2}

  \def\ca{0}
  \def\cb{3}

  \begin{scope}[densely dashed, thin]
    \Vline{\ra}{\rb}{\ca}{$c_1$}
    \Vline{\ra}{\rb}{\cb}{$c_2$}
    \Hline{\ca}{\cb}{\ra}{$r_1$}
    \Hline{\ca}{\cb}{\rb}{$r_2$}
  \end{scope}

  \node [circ] at (\ca*\step, \ra*\step) {};
  \node [circ] at (\cb*\step, \rb*\step) {};


  \vdomino{10}{2}

  \def\ra{3}
  \def\rb{2}

  \def\ca{10}
  \def\cb{10}

  \begin{scope}[densely dashed, thin]
    \Vline{\ra}{\rb}{\ca}{$c_1 = c_2$}
    \Hline{\ca}{\cb}{\ra}{$r_1$}
    \Hline{\ca}{\cb}{\rb}{$r_2$}
  \end{scope}

  \node [circ] at (\ca*\step, \ra*\step) {};
  \node [circ] at (\cb*\step, \rb*\step) {};

\end{tikzpicture}

\caption{Whenever $(r_1,c_1)$ and $(r_2,c_2)$ belong to different tiles (left), Lemma~\ref{lem:PairOfTiles} can be used to upper bound the off-diagonal entry $a_{r_1 r_2} \cdot b_{c_1 c_2}$ of $a \x b$. When both coordinates correspond to the same tile (right), this result cannot be applied directly.}
\label{fig:Pair}
\end{figure}
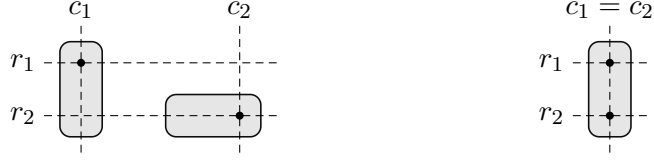

\section{Domino states} \label{sec:Dominoes}

In this section we use the framework introduced earlier to give a lower bound on the error probability of any LOCC measurement for discriminating states from certain bipartite orthonormal product bases known as domino states. This provides an alternative proof of the quantitative separation between LOCC and separable measurements first given in~\cite{IBM} as well as generalizations to states corresponding to other domino-type tilings and a rotated version of the original domino states.

\begin{figure}[!ht]
\centering


\begin{tikzpicture}[
  node distance = 3pt,
  domino/.style = {rectangle, rounded corners = 2mm, draw = black!95, fill = black!10}]

  \def\step{30pt};   
  \def\dshort{24pt}; 
  \def\dlong{55pt};  

  \draw[step = \step, gray, thin] (0,0) grid (3*\step, 3*\step);

  \foreach \i/\j/\h in {1/2.5/3, 2/0.5/0}
    \node[domino, double, minimum height = \dshort, minimum width = \dlong]  at (\i*\step, \j*\step) {};

  \foreach \i/\j in {2.5/2, 0.5/1}
    \node[domino, double, minimum height = \dlong,  minimum width = \dshort] at (\i*\step, \j*\step) {};

  \node[domino, minimum height = \dshort, minimum width = \dshort] at (1.5*\step, 1.5*\step) {};

  \foreach \i in {0,1,2}{
    \node (a\i) at (\i*\step+0.5*\step, 3.5)   {$|\i\rangle$};
    \node (b\i) at (-0.35,-\i*\step+2.5*\step) {$|\i\rangle$};
  };

  \node [above = of a1] {Bob};
  \node [left  = of b1] {Alice};

  \node at (1.5*\step, 1.5*\step) {$1$};
  \node at (1.0*\step, 2.5*\step) {$2,3$};
  \node at (2.0*\step, 0.5*\step) {$4,5$};
  \node at (0.5*\step, 1.0*\step) {$6,7$};
  \node at (2.5*\step, 2.0*\step) {$8,9$};

\end{tikzpicture}

\caption{The tiling induced by states from Equations~(\ref{eq:psi1}--\ref{eq:psi89}).}
\label{fig:Dominoes}
\end{figure}
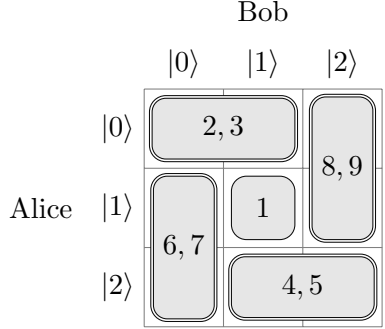

\subsection{Definition}

\newcommand{\thetas}{\theta_1,\theta_2,\theta_3,\theta_4}

The following orthonormal product basis is known as the \emph{domino states}:
\begin{alignat}{3}
  && \ket{\psi_1} = \ket{1} \ket{1}, \label{eq:psi1} \\
  \ket{\psi_{2}} &= \ket{0} \ket{0 + 1}, &&
      &\ket{\psi_{3}} &= \ket{0} \ket{0 - 1}, \\
  \ket{\psi_{4}} &= \ket{2} \ket{1 + 2}, &&
      &\ket{\psi_{5}} &= \ket{2} \ket{1 - 2}, \\
  \ket{\psi_{6}} &= \ket{1 + 2} \ket{0}, &&
      &\ket{\psi_{7}} &= \ket{1 - 2} \ket{0}, \\
  \ket{\psi_{8}} &= \ket{0 + 1} \ket{2}, &&
      &\ket{\psi_{9}} &= \ket{0 - 1} \ket{2}, \label{eq:psi89}
\end{alignat}
where $\ket{i \pm j} \defeq (\ket{i} \pm \ket{j})/\sqrt{2}$.  In \cite{IBM} it was shown that any LOCC protocol for discriminating these states has information deficit at least $5.31 \times 10^{-6}$ (out of $\log_2 9 \approx 3.17$) bits.

In~\cite{IBM} the authors also consider a family of orthonormal product bases, the so-called \emph{rotated domino states}, which are parametrized by four angles $0 \leq \thetas \leq \pi/4$ and are defined as follows:
\begin{alignat}{3}
  && \ket{\psi_1} = \ket{1} \ket{1}, \\
   \ket{\psi_2} &= \ket{0} (\cos \theta_1 \ket{0} + \sin \theta_1 \ket{1}), &&
  &\ket{\psi_3} &= \ket{0} (  - \sin \theta_1 \ket{0} +     \cos \theta_1 \ket{1}), \label{eq:psi23} \\
   \ket{\psi_4} &= \ket{2} (\cos \theta_2 \ket{1} + \sin \theta_2 \ket{2}), &&
  &\ket{\psi_5} &= \ket{2} (  - \sin \theta_2 \ket{1} +     \cos \theta_2 \ket{2}), \\
   \ket{\psi_6} &=         (\cos \theta_3 \ket{1} + \sin \theta_3 \ket{2}) \ket{0}, &&
  &\ket{\psi_7} &=         (  - \sin \theta_3 \ket{1} +     \cos \theta_3 \ket{2}) \ket{0}, \\
   \ket{\psi_8} &=         (\cos \theta_4 \ket{0} + \sin \theta_4 \ket{1}) \ket{2}, &&
  &\ket{\psi_9} &=         (  - \sin \theta_4 \ket{0} +     \cos \theta_4 \ket{1}) \ket{2}.
\end{alignat}
Let $S_3(\thetas)$ denote the rotated domino basis parametrized as above. Then the original domino basis is $S_3 \defeq S_3(\pi/4,\pi/4,\pi/4,\pi/4)$.

Reference~\cite{IBM} shows that states from the domino basis $S_3$ cannot be perfectly discriminated by asymptotic LOCC and conjectures that the same holds for the rotated domino basis $S_3(\thetas)$ for any $0 < \thetas \leq \pi/4$. In the next section we give an alternative proof that quantifies the nonlocality of the original domino states $S_3$ and then adapt the argument to the rotated domino states, thus resolving the conjecture.

\subsection{Nonlocality of the domino states}

To lower bound the nonlocality constant of the domino states $S_3$, we put an upper bound on their rigidity. In other words, we show that measurement operators that only slightly disturb these states are weak (approximately proportional to the identity operator).  The key ingredient of the proof is Lemma~\ref{lem:PairOfTiles} from Section~\ref{sec:Tilings}.

\begin{lemma}
\label{lem:Box}
The domino state basis $S_3$ is $4$-rigid.
\end{lemma}

\begin{proof}
The claimed result can be restated as follows (see Definition~\ref{def:Rigidity}):
\begin{align}
  \Abs{a_{ii} b_{jj} - \frac{1}{9} \tr(a \x b)} &\leq 4 \delta \tr(a \x b),
  \label{eq:Diag} \\
  \abs{a_{ij} b_{kt}} &\leq 4 \delta \tr(a \x b),
  \label{eq:Offdiag}
\end{align}
where $i,j,k,t \in \set{0,1,2}$ and $i \neq j$ or $k \neq t$ in the second equation. First we prove the bound for the diagonal elements and then we proceed to bound the off-diagonal ones.

\partitle{Bounding the diagonal elements:}

We start by bounding the differences of the diagonal elements of matrices $a$ and $b$ separately. Let us rewrite the definition of $\delta$ from Equation~(\ref{eq:delta}) in the case of product states $\ket{\psi_i} = \ket{\alpha_i} \ket{\beta_i}$:
\begin{equation}
  \delta
  = \max_{i \neq j}
     \frac{\abs {\bra{\alpha_i} a \ket{\alpha_j}}}
          {\sqrt{\bra{\alpha_i} a \ket{\alpha_i}
                 \bra{\alpha_j} a \ket{\alpha_j}}} \cdot
     \frac{\abs {\bra{\beta _i} b \ket{\beta _j}}}
          {\sqrt{\bra{\beta _i} b \ket{\beta _i}
                 \bra{\beta _j} b \ket{\beta _j}}}.
  \label{eq:delta ab}
\end{equation}
If we consider the pair of states $\ket{\psi_{2,3}} = \ket{0} \ket{0 \pm 1}$, we get
\begin{align}
  \delta
 &\geq
    \frac{\abs{a_{00}}}{\abs{a_{00}}} \cdot
    \frac{\abs{b_{00}-b_{01}+b_{10}-b_{11}}}
       {\sqrt{(b_{00}+b_{01}+b_{10}+b_{11})
              (b_{00}-b_{01}-b_{10}+b_{11})}} \label{eq:b-diag1} \\
 &=
    \frac{\abs{b_{00}-b_{11} + 2i\Im b_{10}}}
       {\sqrt{(b_{00}+b_{11})^2 - (b_{01}+b_{10})^2}} \\
 &\geq 
    \frac{\abs{b_{00}-b_{11}}}{\abs{b_{00}+b_{11}}} \\
 &\geq 
    \frac{\abs{b_{00}-b_{11}}}{\tr(b)}. \label{eq:b-diag3} 
\end{align}
Note that the cancellation of $\abs{a_{00}}$ is valid since $a_{00} \neq 0$ by the definition of stage~I. Applying a similar argument to the pairs of states from the other three tiles of size $2$, we get that for any $i \in \set{0,2}$,
\begin{align}
  \delta\tr(a) &\geq \abs{a_{11}-a_{ii}} 
  & \text{and} &&
  \delta\tr(b) &\geq \abs{b_{11}-b_{ii}}.
  \label{eq:diag1}
\end{align}
Using these bounds and the triangle inequality, we can bound the difference between the first and last diagonal elements:
\begin{equation}
  \abs{a_{00}-a_{22}} \leq \abs{a_{00}-a_{11}} + \abs{a_{11}-a_{22}}
  \leq 2\delta\tr(a)
  \label{eq:Triangle inequality}
\end{equation}
and similarly $\abs{b_{00}-b_{22}} \leq2 \delta \tr(b)$.

Next, we use the bounds on the differences of the diagonal elements of $a$ and $b$ to bound the differences of the diagonal elements of $a \x b$. For all $i,j,k,t \in \set{0,1,2}$ we have
\begin{align}
  \abs{a_{ii}b_{jj}-a_{kk}b_{tt}}&\leq \label{eq:ab-diag1}
  \abs{a_{ii}b_{jj}-a_{kk}b_{jj}}+\abs{a_{kk}b_{jj}-a_{kk}b_{tt}}\\
  & = \abs{b_{jj}} \cdot \abs{a_{ii}-a_{kk}}+
      \abs{a_{kk}} \cdot \abs{b_{jj}-b_{tt}}\\
  & \leq \abs{b_{jj}} \cdot 2\delta \tr(a)+
         \abs{a_{kk}} \cdot 2\delta \tr(b)\\
  & \leq 4 \delta \tr(a \x b). \label{eq:ab-diag4}
\end{align}

Using this inequality we can obtain the desired bound~(\ref{eq:Diag}) for the diagonal elements: for all $i, j \in \set{0,1,2}$ we have
\begin{align}
  \Abs{a_{ii}b_{jj}-\frac{1}{9}\tr(a \x b)}&= \label{eq:ab-trace1}
  \Abs{a_{ii}b_{jj}-\frac{1}{9}\sum_{k,t\in\set{0,1,2}}a_{kk}b_{tt}}\\
  &\leq \frac{1}{9}\sum_{k,t\in\set{0,1,2}}\abs{a_{ii}b_{jj}-a_{kk}b_{tt}}\\
  &\leq 4 \delta \tr(a \x b). \label{eq:ab-trace3}
\end{align}

\partitle{Bounding the off-diagonal elements:}

From Lemma~\ref{lem:PairOfTiles} we know that $\sqrt{\abs{T_1} \cdot \abs{T_2}} \; \delta \tr(a \x b) \geq \abs{a_{r_1 r_2}} \cdot \abs{b_{c_1 c_2}}$, where $T_1$ and $T_2$ are two distinct tiles and $\abs{T_t}$ is the area of the tile containing $(r_t,c_t)$. For $(r_1, c_1) = (1,1)$ and any $(r_2, c_2) \neq (1,1)$ we get
\begin{equation}
  \sqrt{2} \delta \tr(a \x b)
  \geq \abs{a_{r_1 r_2}} \cdot \abs{b_{c_1 c_2}}.
  \label{eq:Offdiag1}
\end{equation}
Similarly, for any $(r_1, c_1)$ and $(r_2, c_2)$ that belong to distinct tiles of size two we get
\begin{equation}
  2 \delta \tr(a \x b)
  \geq \abs{a_{r_1 r_2}} \cdot \abs{b_{c_1 c_2}}.
  \label{eq:Offdiag2}
\end{equation}
Now it only remains to bound the following four off-diagonal elements (each of which corresponds to one of the four tiles of size $2$):
\begin{equation}
  \abs{a_{00}} \cdot \abs{b_{01}}, \quad
  \abs{a_{01}} \cdot \abs{b_{22}}, \quad
  \abs{a_{22}} \cdot \abs{b_{12}}, \quad
  \abs{a_{12}} \cdot \abs{b_{00}}.
  \label{eq:Hard cases}
\end{equation}

To bound $\abs{a_{00}} \cdot \abs{b_{01}}$, first choose $(r_2,c_2) = (1,0)$ and use Equation~(\ref{eq:Offdiag1}):
\begin{equation}
  \sqrt{2} \delta \tr(a \x b)
  \geq \abs{a_{11}} \cdot \abs{b_{10}}
     = \abs{a_{11}} \cdot \abs{b_{01}}.
\end{equation}
Now it only remains to replace $a_{11}$ by $a_{00}$. Notice from Equation~(\ref{eq:diag1}) that $\delta \tr(a) \geq \abs{a_{11} - a_{00}} \geq \abs{a_{00}} - \abs{a_{11}}$, so
\begin{align}
  \sqrt{2} \delta \tr(a \x b) &\geq 
  \abs{a_{11}} \cdot \abs{b_{01}} \geq
  \bigl(\abs{a_{00}}-\delta\tr(a)\bigr) \cdot \abs{b_{01}} \label{eq:hard1} \\
  & \geq \abs{a_{00}} \cdot \abs{b_{01}} - \delta \tr(a \x b) \label{eq:hard2}
\end{align}
where the last inequality holds since $\abs{b_{01}} \leq \max \set{b_{00}, b_{11}} \leq \tr(b)$ as $b$ is positive semidefinite. After rearranging the previous expression we obtain
\begin{equation}
  (1+\sqrt{2}) \delta \tr(a \x b) \geq \abs{a_{00}} \cdot \abs{b_{01}}.
  \label{eq:Offdiag3}
\end{equation}
By appropriately choosing the value of $(r_2,c_2)$ and using a similar argument, we get the same upper bound for the remaining three off-diagonal elements listed in Equation~(\ref{eq:Hard cases}). Since the constants obtained in bounds~(\ref{eq:Offdiag1}), (\ref{eq:Offdiag2}), and~(\ref{eq:Offdiag3}) satisfy $\max \set{\sqrt{2}, 2, 1+\sqrt{2}} \leq 4$, we have shown that Equation~(\ref{eq:Offdiag}) holds for all off-diagonal elements of $a \x b$.
\end{proof}

Together with Equation~(\ref{eq:eta bound}) this implies that the nonlocality constant for the domino states is $\eta \geq 1/8$. To get an explicit value for the lower bound on the error probability, we use Theorem~\ref{thm:cL} with $n = 9$, $L = 2$, and $c = 4$.

\begin{corollary}
Any LOCC measurement for discriminating the domino states $S_3$ errs with probability
\begin{equation}
  p_{\error} \geq 1.9 \times 10^{-8}.
\end{equation}
\end{corollary}

\subsection{Nonlocality of irreducible domino-type tilings}

Lemma~\ref{lem:Box} can be easily generalized to product bases that are similar to domino states on larger quantum systems.

\begin{restatable}{lemma}{DIMBOX}\label{lem:DimBox}
Let $d_A, d_B \geq 3$ and let $S$ be an orthonormal product basis of $\bip{d_A}{d_B}$. If $S$ induces an irreducible domino-type tiling of diameter $D$ then $S$ is $2D$-rigid (see Section~\ref{sec:Definitions} for terminology).
\end{restatable}

The proof is similar to that of Lemma~\ref{lem:Box} and appears in Appendix~\ref{apx:DimBox}.

To bound the error probability, we use Theorem~\ref{thm:cL} with $n = d_A d_B$, $L = 2$, and $c = 2D$.
\begin{corollary}
Any LOCC measurement for discriminating states from an orthonormal product basis of $\bip{d_A}{d_B}$ that induces an irreducible domino-type tiling of diameter $D$ errs with probability
\begin{equation}
  p_{\error} \geq \frac{1}{216 D^2 (d_A d_B)^5}.
\end{equation}
\end{corollary}

\subsection{Nonlocality of the rotated domino states} \label{sect:Rotated}

The following is an analog of Lemma~\ref{lem:Box} for rotated domino states.

\begin{restatable}{lemma}{ROTBOX}\label{lem:RotBox}
The rotated domino basis $S_3(\thetas)$ is $\frac{C}{\sin 2 \theta}$-rigid where
\begin{equation}
  C \defeq 6 \Bigl( 1 + 6 \sqrt{2} + 2 \sqrt{3 (6 + \sqrt{2})} \Bigr) \leq 114
  \label{eq:C}
\end{equation}
and $\theta \defeq \min \set{\theta_1, \theta_2, \theta_3, \theta_4}$.
\end{restatable}

The proof appears in Appendix~\ref{apx:RotBox}.

Again, we use Theorem~\ref{thm:cL} to lower bound the error probability. Here the parameters are $n = 9$, $L = 2$, and $c = 114/\sin(2 \theta)$.

\begin{corollary}
\label{cor:RotBox}
Any LOCC measurement for discriminating $S_3(\thetas)$, the set of rotated domino states, errs with probability
\begin{equation}
  p_{\error}
  \geq 2.4 \times 10^{-11} \sin^2 (2 \theta),
\end{equation}
where $\theta \defeq \min \set{\theta_1, \theta_2, \theta_3, \theta_4}$.
\end{corollary}

Note that as $\theta$ approaches zero, the rigidity bound tends to infinity and the bound on the error probability goes to zero. As the original domino basis is transformed continuously to the standard basis, the nonlocality decreases to zero. Moreover, since any orthonormal product basis of $\bips{3}$ is equivalent to $S_3(\thetas)$ (up to local unitary transformations) for some angles $\theta_i$ \cite{Feng3}, Corollary~\ref{cor:RotBox} effectively covers all product bases of $\bips{3}$.

\section{Limitations of the framework} \label{sec:Limitations}

\subsection{Dependence of the nonlocality constant on \texorpdfstring{$n$}{n}} \label{subsec:SmallBound}

Recall that in Theorem~\ref{thm:eta} we established the lower bound $p_{\error} \geq \frac{2}{27} \frac{\eta^2}{n^5}$ on the error probability, where $\eta$ is the nonlocality constant and $n$ is the number of states. Intuitively it seems that it should be possible to prove a stronger lower bound on $p_{\error}$ as $n$ increases. However, to lower bound $p_{\error}$ by a fixed constant in \emph{any} dimension using our framework, one would have to prove a lower bound on $\eta$ that increases with $n$. 

Let us consider the problem of discriminating orthonormal product states. In the next lemma we show that it is not possible to obtain such strong error bounds using our framework in its present form. We do this by proving a fixed upper bound on the nonlocality constant in any dimension.

\begin{lemma}
Let $S$ be a set of orthonormal product states in $\bip{d_A}{d_B}$. The nonlocality constant of $S$ satisfies $\eta \leq 2$.
\end{lemma}

\begin{proof}
Let $n = \abs{S}$ and $\ket{\psi_i} = \ket{\alpha_i} \ket{\beta_i}$. Fix some small $\epsilon > 0$, choose any $i \in [n]$, and define
\begin{align}
  a &= \ketbra{\alpha_i}{\alpha_i} + \epsilon I_{d_A}, &
  b &= \ketbra{\beta _i}{\beta _i} + \epsilon I_{d_B}.
\end{align}
Note that $a$ and $b$ have full rank and are positive semidefinite. We can easily check that
\begin{align}
  \tr(a) &= 1 + \epsilon d_A, &
  \tr(b) &= 1 + \epsilon d_B, &
  \max_{k \in [n]} \bra{\psi_k} (a \x b) \ket{\psi_k} &= (1 + \epsilon)^2.
\end{align}
Using these observations together with the definition of $\eta$ in Equation~(\ref{eq:eta}), we get
\begin{align}
  \eta \biggl( \frac{(1+\epsilon)^2}{(1 + \epsilon d_A)(1 + \epsilon d_B)} - \frac{1}{n} \biggr)
  & \leq \eta \biggl(
         \frac{      \max_{k \in [n]} \bra{\psi_k} (a \x b) \ket{\psi_k}}
              {\;\;\;\sum_{j \in [n]} \bra{\psi_j} (a \x b) \ket{\psi_j}} - \frac{1}{n}
       \biggr)\\
  & =\delta_S(a \x b) \\
  &\leq 1,
\end{align}
where the last inequality follows directly from Definition~\ref{def:delta}. As $\epsilon \rightarrow 0$, the left-hand side goes to $\eta(1-\frac{1}{n})$. We can choose $\epsilon$ arbitrarily small, so $\eta(1 - \frac{1}{n}) \leq 1$ and thus $\eta \leq \frac{n}{n-1} = 1 + \frac{1}{n-1} \leq 2$ since $n \geq 2$.
\end{proof}

\subsection{Comparison to the result of Kleinmann, Kampermann, and Bru\ss} \label{sec:KKB}

The main application of the framework introduced in this paper is to show the impossibility of asymptotically discriminating a set of states $S$ with LOCC. We do this by showing that the nonlocality constant of $S$ is strictly positive. In other words, the nonlocality constant being zero is a necessary condition for the sates in $S$ to be asymptotically distinguishable with LOCC. Another necessary condition is presented in recent work of Kleinmann, Kampermann, and Bru\ss:

\begin{theorem*}[\cite{KKB}]
Let $S=\set{\rho_1,\dotsc,\rho_n}$ be a set of states such that $\bigcap_i \ker \rho_i$ does not contain any nonzero product vector. 
Then $S$ can be asymptotically discriminated with LOCC only if for all $\chi$ with $1/n\leq \chi \leq 1$ there exists a positive semidefinite product operator $E$ satisfying all of the following:
\begin{enumerate}
\item $\sum_i\tr(E\rho_i)=1$, 
\item $\max_i\tr(E\rho_i)=\chi$,
\item $\tr(E\rho_i E \rho_j)=0$ for any $i\neq j$.
\end{enumerate}
\end{theorem*}

One should note however that in contrast to the above \emph{qualitative} result, our framework can be applied to any set of orthogonal pure states (with no restriction on $\bigcap_i \ker \rho_i$) and can be used to obtain explicit lower bounds on the error probability. It is an open question whether our necessary condition (``the nonlocality constant of $S$ is zero'') or that of the above theorem is also sufficient. The lemma below shows that if our necessary condition is also sufficient then so is that of \cite{KKB}.  

\begin{lemma}
Let $S=\set{\ket{\psi_i}}_{i\in[n]}$ be a set of orthogonal pure states such that $\bigcap_i \ker (\ketbras{\psi_i})$ does not contain any nonzero product vector. If for all $\chi$ with $1/n\leq \chi \leq 1$ there exists a positive semidefinite product operator $E$ satisfying conditions 1--3 from the above theorem, then the nonlocality constant $\eta$ of $S$ is zero.
\end{lemma}

\begin{proof}
Consider $\chi\in(\frac{1}{n},\frac{1}{n-1})$ and a positive semidefinite product operator $E_\chi$ satisfying conditions 1--3. Conditions 1 and 2 imply that $\bra{\psi_i}E_\chi\ket{\psi_i}>0$ thus making $\delta_S(E_\chi)$ well defined (see Definition~\ref{def:delta}). Moreover, by condition 3 we have that $\abs{\bra{\psi_i}E_\chi\ket{\psi_j}}^2=0$ for all $i\neq j$. Hence $\delta_S(E_\chi)=0$ according to Definition~\ref{def:delta}. Finally, from conditions 1 and 2, we get that
\begin{equation}
 \frac{      \max_{i} \bra{\psi_i} E_\chi \ket{\psi_i}}
       {\;\;\;\sum_{j} \bra{\psi_j} E_\chi \ket{\psi_j}}
  = \frac{\max_i\tr(E\rho_i)}{\sum_j\tr(E\rho_j)}
  = \chi.
\end{equation}
Using these observations we can rewrite Equation (\ref{eq:eta}) in the definition of $\eta$ as
\begin{equation}
  \eta \Bigl( \chi - \frac{1}{n} \Bigr) \leq 0.
\end{equation}
Since $\chi > \frac{1}{n}$ it follows from the above inequality that $\eta=0$.
\end{proof}

\section{Discussion and open problems} \label{sec:Conclusions}

We have developed a framework for quantifying the hardness of distinguishing sets of bipartite pure states with LOCC. Using this framework, we proved lower bounds on the error probability of distinguishing several sets of states, as summarized in Table~\ref{tab:Summary}.

\begin{table}[!ht]
\begin{center}
\newcommand{\rh}[2]{\rule{0pt}{#1}\\[#2]}
\begin{tabular}{l|c|c|c}
\textit{Set of states} & $c$ & $\eta$ & $p_{\error}$
\rh{0.5cm}{0.1cm} \hline \hline
Domino states                    & $4$
                                 & $\displaystyle \frac{1}{8}$
                                 & $1.96 \times 10^{-8}$
\rh{0.7cm}{0.3cm} \hline
Domino-type states               & $2D$
                                 & $\displaystyle \frac{1}{4D}$
                                 & $\displaystyle \frac{1}{216 D^2 (d_A d_B)^5}$
\rh{0.7cm}{0.3cm} \hline
$\theta$-rotated domino states   & $\displaystyle \frac{114}{\sin 2 \theta}$
                                 & $\displaystyle \frac{\sin 2 \theta}{227}$
                                 & $2.41 \times 10^{-11} \sin^2 (2 \theta)$
\rh{0.7cm}{0.3cm}
\end{tabular}
\caption{Rigidity $c$ and lower bounds on the nonlocality constant $\eta$ and error probability $p_{\error}$ for various states.}
\label{tab:Summary}
\end{center}
\end{table}

This work raises several open problems. While we were able to lower bound the nonlocality constant $\eta$ in many cases, it could be useful to develop more generic approaches to computing or lower bounding this quantity. We are also interested in applying our method to other sets of states. For example, we would like to apply the method when $S$ is an incomplete orthonormal set (e.g., the domino basis with the middle tile omitted) or a product basis with tiles of size larger than two (see Fig.~\ref{fig:Tilings} for concrete examples of such tilings where no bounds on $p_{\error}$ are known). It is unknown whether there exists a set $S$ of 2-qubit states that can be perfectly discriminated with separable operations, but for which any LOCC protocol has $p_{\error}(S) > 0$ (see \cite{DuanSep} for all possible candidate sets). Finally, it would be interesting to consider random product bases, since this would tell us how generic the phenomenon of nonlocality without entanglement is.

\begin{figure}[!ht]
\centering


\def\step{20pt} 

\begin{tikzpicture}[
  domino/.style = {rectangle, rounded corners = 0.2*\step, draw = black!95, fill = black!10},
  circ/.style = {circle, draw = black, fill = black, inner sep = 0mm, minimum size = 0.13*\step},
  gridlines/.style = {gray, semithick}
]

  \newcommand{\domino}[4]{
    \node[semithick, domino,
          minimum height = #3*\step + 0.8*\step,
          minimum width  = #4*\step + 0.8*\step] at (#1*\step, #2*\step) {};
    \foreach \i in {0,...,#4}{
      \foreach \j in {0,...,#3}{
        \node at (\i*\step+#1*\step-#4*\step/2,
                  \j*\step+#2*\step-#3*\step/2) [circ] {};
      }
    };
  }

  \newcommand{\sizeone}[2]{
    \pgfmathparse{#1+0.5};
    \let\x = \pgfmathresult;
    \pgfmathparse{#2+0.5};
    \let\y = \pgfmathresult;
    \domino{\x}{\y}{0}{0}
  }

  \newcommand{\sizetwo}[4]{
    \pgfmathparse{#1+1-0.5*#3};
    \let\x = \pgfmathresult;
    \pgfmathparse{#2+1-0.5*#4};
    \let\y = \pgfmathresult;
    \domino{\x}{\y}{#3}{#4}
    \draw[semithick]
         (\x*\step - #4*\step/2, \y*\step - #3*\step/2) --
         (\x*\step + #4*\step/2, \y*\step + #3*\step/2);
  }

  \newcommand{\sizefour}[2]{
    \pgfmathparse{#1+1.0};
    \let\x = \pgfmathresult;
    \pgfmathparse{#2+1.0};
    \let\y = \pgfmathresult;
    \domino{\x}{\y}{1}{1}
    \draw[semithick]
         (#1*\step + 1*\step/2, #2*\step + 1*\step/2) --
         (#1*\step + 3*\step/2, #2*\step + 1*\step/2) --
         (#1*\step + 3*\step/2, #2*\step + 3*\step/2) --
         (#1*\step + 1*\step/2, #2*\step + 3*\step/2) -- cycle;
  }

  \newcommand{\tiling}[3]{
    \begin{scope}[fill = white]
      \fill[clip] (0,0) rectangle (#1*\step, #2*\step);
      \draw[step = \step, gridlines] (0,0) grid (#1*\step, #2*\step);
      #3
    \end{scope}
    \draw[gridlines] (0,0) -- (#1*\step, 0) -- (#1*\step, #2*\step) -- (0, #2*\step) -- cycle;
  }


  \tiling{5}{5}{
    \sizeone{0}{3}
    \sizeone{1}{0}
    \sizeone{3}{4}
    \sizeone{4}{1}
    \sizeone{2}{2}

    \sizefour{0}{1}
    \sizefour{2}{0}
    \sizefour{3}{2}
    \sizefour{1}{3}

    \sizefour{-1}{-1}
    \sizefour{-1}{ 4}
    \sizefour{ 4}{-1}
    \sizefour{ 4}{ 4}
  }


  \begin{scope}[shift = {(-7*\step,1*\step)}]
  \tiling{4}{3}{
    \sizetwo{0}{0}{0}{1}
    \sizetwo{1}{1}{0}{1}
    \sizetwo{2}{2}{0}{1}
    \sizetwo{0}{1}{1}{0}
    \sizetwo{3}{0}{1}{0}
  }
  \end{scope}

\end{tikzpicture}

\caption{Tilings corresponding to an incomplete orthonormal set in $\bip{3}{4}$ (left) and a product basis of $\bip{5}{5}$ with larger tiles (right). On the right, the tiles of size four are induced by states of the form $\ket{\pm} \ket{\pm}$ and one of the tiles corresponds to the four corners of the grid.}
\label{fig:Tilings}
\end{figure}
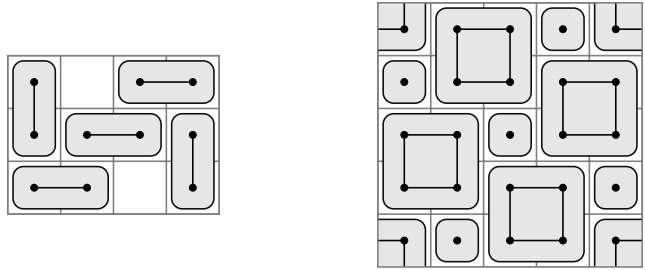

We discussed some limitations of our framework in Section~\ref{sec:Limitations}, but we would like to better understand how broadly the framework can be applied. In particular, can it always be used to obtain a lower bound on $p_{\error}$ whenever such a bound exists? For example, from Section~\ref{sect:Rotated} we know that the answer to this question is ``yes'' for orthonormal product bases on two qutrits.

Finally, the gaps between the classes of separable and LOCC operations exhibited by our framework are rather small (see Table~\ref{tab:Summary}). One cannot hope to do significantly better within our framework, as shown in Section~\ref{subsec:SmallBound}. Is this due to limitations of our framework or because orthonormal product states in general can be discriminated well by LOCC?

Along these lines, a major open question raised by our work is the following: does there exist a sequence $S_1, S_2, S_3, \dotsc$ of sets of orthonormal product states such that
\begin{equation*}
  \lim_{l \to \infty} p^{LOCC}_{\error}(S_l) = 1?
\end{equation*}
Existence of such a sequence would give a strong separation between the classes of separable and LOCC measurements.
Note that the local standard basis measurement followed by guessing gives the correct answer with probability at least $1 / L_l$, where $L_l$ is the maximum number of states within a tile in the tiling induced by $S_l$. Thus for any such sequence, the value of $L_l$ must grow with $l$. In particular, the number of states (and hence the local dimensions) must also grow with $l$.

\section{Acknowledgements}

We thank Dagmar Bru\ss{}, Eric Chitambar, Sarah Croke, Oleg Gittsovich, Tsuyoshi Ito, Hermann Kampermann, Matthias Kleinmann, Will Matthews, Rajat Mittal, Marco Piani, David Roberson, Graeme Smith, and John Watrous for helpful discussions. This research was funded by CRC, CFI, CIFAR, MITACS, NSERC, ORF, the Ontario Ministry of Research and Innovation, and the US ARO/DTO.

\bibliographystyle{alphaurl}

\begin{thebibliography}{HMM{\etalchar{+}}06}

\bibitem[BBKW09]{BBKW}
Somshubhro Bandyopadhyay, Gilles Brassard, Shelby Kimmel, and William~K.
  Wootters.
\newblock Entanglement cost of nonlocal measurements.
\newblock {\em Phys. Rev. A}, 80:012313, Jul 2009.
\newblock \href {http://arxiv.org/abs/0809.2264} {\path{arXiv:0809.2264}},
  \href {http://dx.doi.org/10.1103/PhysRevA.80.012313}
  {\path{doi:10.1103/PhysRevA.80.012313}}.

\bibitem[BDF{\etalchar{+}}99]{IBM}
Charles~H. Bennett, David~P. DiVincenzo, Christopher~A. Fuchs, Tal Mor, Eric
  Rains, Peter~W. Shor, John~A. Smolin, and William~K. Wootters.
\newblock Quantum nonlocality without entanglement.
\newblock {\em Phys. Rev. A}, 59:1070--1091, Feb 1999.
\newblock \href {http://arxiv.org/abs/quant-ph/9804053}
  {\path{arXiv:quant-ph/9804053}}, \href
  {http://dx.doi.org/10.1103/PhysRevA.59.1070}
  {\path{doi:10.1103/PhysRevA.59.1070}}.

\bibitem[BDM{\etalchar{+}}99]{BDMSST}
Charles~H. Bennett, David~P. DiVincenzo, Tal Mor, Peter~W. Shor, John~A.
  Smolin, and Barbara~M. Terhal.
\newblock Unextendible product bases and bound entanglement.
\newblock {\em Phys. Rev. Lett.}, 82:5385--5388, Jun 1999.
\newblock \href {http://arxiv.org/abs/quant-ph/9808030}
  {\path{arXiv:quant-ph/9808030}}, \href
  {http://dx.doi.org/10.1103/PhysRevLett.82.5385}
  {\path{doi:10.1103/PhysRevLett.82.5385}}.

\bibitem[BT10]{Discord}
Aharon Brodutch and Daniel~R. Terno.
\newblock Quantum discord, local operations, and {M}axwell's demons.
\newblock {\em Phys. Rev. A}, 81:062103, Jun 2010.
\newblock \href {http://arxiv.org/abs/1002.4913} {\path{arXiv:1002.4913}},
  \href {http://dx.doi.org/10.1103/PhysRevA.81.062103}
  {\path{doi:10.1103/PhysRevA.81.062103}}.

\bibitem[CCL11]{ChitambarSep}
Wei Cui, Eric Chitambar, and Hoi-Kwong Lo.
\newblock Increasing entanglement by separable operations and new monotones for
  {W}-type entanglement.
\newblock 2011.
\newblock \href {http://arxiv.org/abs/1106.1208} {\path{arXiv:1106.1208}}.

\bibitem[Che04]{Chefles}
Anthony Chefles.
\newblock Condition for unambiguous state discrimination using local operations
  and classical communication.
\newblock {\em Phys. Rev. A}, 69:050307, May 2004.
\newblock \href {http://dx.doi.org/10.1103/PhysRevA.69.050307}
  {\path{doi:10.1103/PhysRevA.69.050307}}.

\bibitem[CL03]{ChenLi-Criterion}
Ping-Xing Chen and Cheng-Zu Li.
\newblock Orthogonality and distinguishability: Criterion for local
  distinguishability of arbitrary orthogonal states.
\newblock {\em Phys. Rev. A}, 68:062107, Dec 2003.
\newblock \href {http://arxiv.org/abs/quant-ph/0209048}
  {\path{arXiv:quant-ph/0209048}}, \href
  {http://dx.doi.org/10.1103/PhysRevA.68.062107}
  {\path{doi:10.1103/PhysRevA.68.062107}}.

\bibitem[CL04]{ChenLi-ProdBasis}
Ping-Xing Chen and Cheng-Zu Li.
\newblock Distinguishing the elements of a full product basis set needs only
  projective measurements and classical communication.
\newblock {\em Phys. Rev. A}, 70:022306, Aug 2004.
\newblock \href {http://dx.doi.org/10.1103/PhysRevA.70.022306}
  {\path{doi:10.1103/PhysRevA.70.022306}}.

\bibitem[Coh07]{Cohen07}
Scott~M. Cohen.
\newblock Local distinguishability with preservation of entanglement.
\newblock {\em Phys. Rev. A}, 75:052313, May 2007.
\newblock \href {http://dx.doi.org/10.1103/PhysRevA.75.052313}
  {\path{doi:10.1103/PhysRevA.75.052313}}.

\bibitem[Coh08]{Cohen-entanglement}
Scott~M. Cohen.
\newblock Understanding entanglement as resource: Locally distinguishing
  unextendible product bases.
\newblock {\em Phys. Rev. A}, 77:012304, Jan 2008.
\newblock \href {http://arxiv.org/abs/0708.2396} {\path{arXiv:0708.2396}},
  \href {http://dx.doi.org/10.1103/PhysRevA.77.012304}
  {\path{doi:10.1103/PhysRevA.77.012304}}.

\bibitem[Coh11]{Cohen}
Scott~M. Cohen.
\newblock When a quantum measurement can be implemented locally, and when it
  cannot.
\newblock {\em Phys. Rev. A}, 84:052322, Nov 2011.
\newblock \href {http://arxiv.org/abs/0912.1607} {\path{arXiv:0912.1607}},
  \href {http://dx.doi.org/10.1103/PhysRevA.84.052322}
  {\path{doi:10.1103/PhysRevA.84.052322}}.

\bibitem[Cro12]{Croke}
Sarah Croke.
\newblock There is no non-local information in a single qubit.
\newblock In {\em APS Meeting Abstracts}, page 30011, Feb 2012.

\bibitem[CY01]{ChenYang-Multipartite}
Yi-Xin Chen and Dong Yang.
\newblock Optimal conclusive discrimination of two nonorthogonal pure product
  multipartite states through local operations.
\newblock {\em Phys. Rev. A}, 64:064303, Nov 2001.
\newblock \href {http://dx.doi.org/10.1103/PhysRevA.64.064303}
  {\path{doi:10.1103/PhysRevA.64.064303}}.

\bibitem[CY02]{ChenYang-Entangled}
Yi-Xin Chen and Dong Yang.
\newblock Optimally conclusive discrimination of nonorthogonal entangled states
  by local operations and classical communications.
\newblock {\em Phys. Rev. A}, 65:022320, Jan 2002.
\newblock \href {http://dx.doi.org/10.1103/PhysRevA.65.022320}
  {\path{doi:10.1103/PhysRevA.65.022320}}.

\bibitem[DFJY07]{DuanUB}
Runyao Duan, Yuan Feng, Zhengfeng Ji, and Mingsheng Ying.
\newblock Distinguishing arbitrary multipartite basis unambiguously using local
  operations and classical communication.
\newblock {\em Phys. Rev. Lett.}, 98:230502, Jun 2007.
\newblock \href {http://arxiv.org/abs/quant-ph/0612034}
  {\path{arXiv:quant-ph/0612034}}, \href
  {http://dx.doi.org/10.1103/PhysRevLett.98.230502}
  {\path{doi:10.1103/PhysRevLett.98.230502}}.

\bibitem[DFXY09]{DuanSep}
Runyao Duan, Yuan Feng, Yu~Xin, and Mingsheng Ying.
\newblock Distinguishability of quantum states by separable operations.
\newblock {\em IEEE Trans. Inf. Theor.}, 55:1320--1330, March 2009.
\newblock \href {http://arxiv.org/abs/0705.0795} {\path{arXiv:0705.0795}},
  \href {http://dx.doi.org/10.1109/TIT.2008.2011524}
  {\path{doi:10.1109/TIT.2008.2011524}}.

\bibitem[DMS{\etalchar{+}}03]{DMSST}
David~P. DiVincenzo, Tal Mor, Peter~W. Shor, John~A. Smolin, and Barbara~M.
  Terhal.
\newblock Unextendible product bases, uncompletable product bases and bound
  entanglement.
\newblock {\em Communications in Mathematical Physics}, 238:379--410, 2003.
\newblock \href {http://arxiv.org/abs/quant-ph/9908070}
  {\path{arXiv:quant-ph/9908070}}, \href
  {http://dx.doi.org/10.1007/s00220-003-0877-6}
  {\path{doi:10.1007/s00220-003-0877-6}}.

\bibitem[DR04]{Rinaldis}
Sergio De~Rinaldis.
\newblock Distinguishability of complete and unextendible product bases.
\newblock {\em Phys. Rev. A}, 70:022309, Aug 2004.
\newblock \href {http://arxiv.org/abs/quant-ph/0304027}
  {\path{arXiv:quant-ph/0304027}}, \href
  {http://dx.doi.org/10.1103/PhysRevA.70.022309}
  {\path{doi:10.1103/PhysRevA.70.022309}}.

\bibitem[DXY10]{DuanXinYing}
Runyao Duan, Yu~Xin, and Mingsheng Ying.
\newblock Locally indistinguishable subspaces spanned by three-qubit
  unextendible product bases.
\newblock {\em Phys. Rev. A}, 81:032329, Mar 2010.
\newblock \href {http://arxiv.org/abs/0708.3559} {\path{arXiv:0708.3559}},
  \href {http://dx.doi.org/10.1103/PhysRevA.81.032329}
  {\path{doi:10.1103/PhysRevA.81.032329}}.

\bibitem[Fan04]{Fan}
Heng Fan.
\newblock Distinguishability and indistinguishability by local operations and
  classical communication.
\newblock {\em Phys. Rev. Lett.}, 92:177905, Apr 2004.
\newblock \href {http://arxiv.org/abs/quant-ph/0311026}
  {\path{arXiv:quant-ph/0311026}}, \href
  {http://dx.doi.org/10.1103/PhysRevLett.92.177905}
  {\path{doi:10.1103/PhysRevLett.92.177905}}.

\bibitem[FS09]{Feng3}
Yuan Feng and Yaoyun Shi.
\newblock Characterizing locally indistinguishable orthogonal product states.
\newblock {\em IEEE Trans. Inf. Theor.}, 55:2799--2806, June 2009.
\newblock \href {http://arxiv.org/abs/0707.3581} {\path{arXiv:0707.3581}},
  \href {http://dx.doi.org/10.1109/TIT.2009.2018330}
  {\path{doi:10.1109/TIT.2009.2018330}}.

\bibitem[GKR{\etalchar{+}}01]{Ghosh-Bell}
Sibasish Ghosh, Guruprasad Kar, Anirban Roy, Aditi Sen(De), and Ujjwal Sen.
\newblock Distinguishability of {B}ell states.
\newblock {\em Phys. Rev. Lett.}, 87:277902, Dec 2001.
\newblock \href {http://arxiv.org/abs/quant-ph/0106148}
  {\path{arXiv:quant-ph/0106148}}, \href
  {http://dx.doi.org/10.1103/PhysRevLett.87.277902}
  {\path{doi:10.1103/PhysRevLett.87.277902}}.

\bibitem[GKRS04]{Ghosh-MaxEnt}
Sibasish Ghosh, Guruprasad Kar, Anirban Roy, and Debasis Sarkar.
\newblock Distinguishability of maximally entangled states.
\newblock {\em Phys. Rev. A}, 70:022304, Aug 2004.
\newblock \href {http://arxiv.org/abs/quant-ph/0205105}
  {\path{arXiv:quant-ph/0205105}}, \href
  {http://dx.doi.org/10.1103/PhysRevA.70.022304}
  {\path{doi:10.1103/PhysRevA.70.022304}}.

\bibitem[GV01]{Groisman}
Berry Groisman and Lev Vaidman.
\newblock Nonlocal variables with product-state eigenstates.
\newblock {\em Journal of Physics A: Mathematical and General}, 34(35):6881,
  2001.
\newblock \href {http://arxiv.org/abs/quant-ph/0103084}
  {\path{arXiv:quant-ph/0103084}}, \href
  {http://dx.doi.org/10.1088/0305-4470/34/35/313}
  {\path{doi:10.1088/0305-4470/34/35/313}}.

\bibitem[Hel76]{Helstrom}
Carl~W. Helstrom.
\newblock {\em Quantum detection and estimation theory}.
\newblock Mathematics in science and engineering. Academic Press, 1976.
\newblock URL: \url{http://books.google.ca/books?id=Ne3iT_QLcsMC&pg=PA113}.

\bibitem[HM03]{HilleryMimih}
Mark Hillery and Jihane Mimih.
\newblock Distinguishing two-qubit states using local measurements and
  restricted classical communication.
\newblock {\em Phys. Rev. A}, 67:042304, Apr 2003.
\newblock \href {http://arxiv.org/abs/quant-ph/0210179}
  {\path{arXiv:quant-ph/0210179}}, \href
  {http://dx.doi.org/10.1103/PhysRevA.67.042304}
  {\path{doi:10.1103/PhysRevA.67.042304}}.

\bibitem[HMM{\etalchar{+}}06]{HMMOV}
Masahito Hayashi, Damian Markham, Mio Murao, Masaki Owari, and Shashank
  Virmani.
\newblock Bounds on multipartite entangled orthogonal state discrimination
  using local operations and classical communication.
\newblock {\em Phys. Rev. Lett.}, 96:040501, Feb 2006.
\newblock \href {http://arxiv.org/abs/quant-ph/0506170}
  {\path{arXiv:quant-ph/0506170}}, \href
  {http://dx.doi.org/10.1103/PhysRevLett.96.040501}
  {\path{doi:10.1103/PhysRevLett.96.040501}}.

\bibitem[HSSH03]{Horodecki}
Micha\l{} Horodecki, Aditi Sen(De), Ujjwal Sen, and Karol Horodecki.
\newblock Local indistinguishability: More nonlocality with less entanglement.
\newblock {\em Phys. Rev. Lett.}, 90:047902, Jan 2003.
\newblock \href {http://arxiv.org/abs/quant-ph/0301106}
  {\path{arXiv:quant-ph/0301106}}, \href
  {http://dx.doi.org/10.1103/PhysRevLett.90.047902}
  {\path{doi:10.1103/PhysRevLett.90.047902}}.

\bibitem[JCY05]{JiCaoYing}
Zhengfeng Ji, Hongen Cao, and Mingsheng Ying.
\newblock Optimal conclusive discrimination of two states can be achieved
  locally.
\newblock {\em Phys. Rev. A}, 71:032323, Mar 2005.
\newblock \href {http://arxiv.org/abs/quant-ph/0407120}
  {\path{arXiv:quant-ph/0407120}}, \href
  {http://dx.doi.org/10.1103/PhysRevA.71.032323}
  {\path{doi:10.1103/PhysRevA.71.032323}}.

\bibitem[KKB11]{KKB}
Matthias Kleinmann, Hermann Kampermann, and Dagmar Bru\ss{}.
\newblock Asymptotically perfect discrimination in the
  local-operation-and-classical-communication paradigm.
\newblock {\em Phys. Rev. A}, 84:042326, Oct 2011.
\newblock \href {http://arxiv.org/abs/1105.5132} {\path{arXiv:1105.5132}},
  \href {http://dx.doi.org/10.1103/PhysRevA.84.042326}
  {\path{doi:10.1103/PhysRevA.84.042326}}.

\bibitem[Koa09]{Koashi09}
Masato Koashi.
\newblock On the irreversibility of measurements of correlations.
\newblock {\em Journal of Physics: Conference Series}, 143(1):012007, 2009.
\newblock \href {http://dx.doi.org/10.1088/1742-6596/143/1/012007}
  {\path{doi:10.1088/1742-6596/143/1/012007}}.

\bibitem[KTYI07]{Koashi07}
Masato Koashi, Fumitaka Takenaga, Takashi Yamamoto, and Nobuyuki Imoto.
\newblock Quantum nonlocality without entanglement in a pair of qubits.
\newblock 2007.
\newblock \href {http://arxiv.org/abs/0709.3196} {\path{arXiv:0709.3196}}.

\bibitem[Nat05]{Nathanson}
Michael Nathanson.
\newblock Distinguishing bipartitite orthogonal states using {LOCC}: Best and
  worst cases.
\newblock {\em Journal of Mathematical Physics}, 46(6):062103, 2005.
\newblock \href {http://arxiv.org/abs/quant-ph/0411110}
  {\path{arXiv:quant-ph/0411110}}, \href {http://dx.doi.org/10.1063/1.1914731}
  {\path{doi:10.1063/1.1914731}}.

\bibitem[NC06]{NisetCerf}
Julien Niset and Nicolas~J. Cerf.
\newblock Multipartite nonlocality without entanglement in many dimensions.
\newblock {\em Phys. Rev. A}, 74:052103, Nov 2006.
\newblock \href {http://arxiv.org/abs/quant-ph/0606227}
  {\path{arXiv:quant-ph/0606227}}, \href
  {http://dx.doi.org/10.1103/PhysRevA.74.052103}
  {\path{doi:10.1103/PhysRevA.74.052103}}.

\bibitem[VSPM01]{Virmani}
Shashank Virmani, Massimiliano~F. Sacchi, Martin~B. Plenio, and Damian Markham.
\newblock Optimal local discrimination of two multipartite pure states.
\newblock {\em Physics Letters A}, 288(2):62--68, 2001.
\newblock \href {http://arxiv.org/abs/quant-ph/0102073}
  {\path{arXiv:quant-ph/0102073}}, \href
  {http://dx.doi.org/10.1016/S0375-9601(01)00484-4}
  {\path{doi:10.1016/S0375-9601(01)00484-4}}.

\bibitem[Wat05]{Watrous}
John Watrous.
\newblock Bipartite subspaces having no bases distinguishable by local
  operations and classical communication.
\newblock {\em Phys. Rev. Lett.}, 95:080505, Aug 2005.
\newblock \href {http://dx.doi.org/10.1103/PhysRevLett.95.080505}
  {\path{doi:10.1103/PhysRevLett.95.080505}}.

\bibitem[WH02]{Walgate}
Jonathan Walgate and Lucien Hardy.
\newblock Nonlocality, asymmetry, and distinguishing bipartite states.
\newblock {\em Phys. Rev. Lett.}, 89:147901, Sep 2002.
\newblock \href {http://arxiv.org/abs/quant-ph/0202034}
  {\path{arXiv:quant-ph/0202034}}, \href
  {http://dx.doi.org/10.1103/PhysRevLett.89.147901}
  {\path{doi:10.1103/PhysRevLett.89.147901}}.

\bibitem[WSHV00]{Walgate-Multi}
Jonathan Walgate, Anthony~J. Short, Lucien Hardy, and Vlatko Vedral.
\newblock Local distinguishability of multipartite orthogonal quantum states.
\newblock {\em Phys. Rev. Lett.}, 85:4972--4975, Dec 2000.
\newblock \href {http://arxiv.org/abs/quant-ph/0007098}
  {\path{arXiv:quant-ph/0007098}}, \href
  {http://dx.doi.org/10.1103/PhysRevLett.85.4972}
  {\path{doi:10.1103/PhysRevLett.85.4972}}.

\end{thebibliography}
\newcommand{\etalchar}[1]{$^{#1}$}

\appendix

\section{Rigidity of domino-type states (Lemma~\ref{lem:DimBox})} \label{apx:DimBox}

\DIMBOX*

\begin{proof}
We mimic the proof of Lemma~\ref{lem:Box} and make the appropriate generalizations when necessary. We want to show that
\begin{align}
  \Abs{a_{ii} b_{jj} - \frac{1}{d_A d_B} \tr(a \x b)} &\leq 2D \delta \tr(a \x b),
  \label{eq:DimDiag} \\
  \abs{a_{ij} b_{kt}} &\leq 2D \delta \tr(a \x b),
  \label{eq:DimOffdiag}
\end{align}
where $i \neq j$ or $k \neq t$ in the second inequality.

\partitle{Bounding the diagonal elements:}

Using the calculation in Equations~\mbox{(\ref{eq:b-diag1}--\ref{eq:b-diag3})} we can bound the difference of diagonal entries of $a$ and $b$. Whenever there is a $2 \times 1$ tile that connects rows $i$ and $j$, we get that
\begin{equation}
  \abs{a_{ii}-a_{jj}} \leq \delta \tr(a).
\end{equation}
A similar equation holds for $b$ whenever there is a $1 \times 2$ tile that connects columns $i$ and $j$.

Since $T$ is irreducible, the row graph of $T$ is connected. Moreover, any two vertices of this graph are connected by a path of length at most $D$. We apply the triangle inequality along this path in the same way as in Equation~(\ref{eq:Triangle inequality}). After at most $D-1$ repetitions we get that for any $i$ and $j$,
\begin{equation}
  \abs{a_{ii}-a_{jj}} \leq D \delta \tr(a).
  \label{eq:a-diag}
\end{equation}
A similar equation holds for $b$. When we repeat the calculation in Equations~\mbox{(\ref{eq:ab-diag1}--\ref{eq:ab-diag4})}, we get that for any $i,j,k,t$,
\begin{equation}
  \abs{a_{ii}b_{jj}-a_{kk}b_{tt}} \leq 2D \delta \tr(a \x b).
\end{equation}
Finally, we repeat the calculation in Equations~\mbox{(\ref{eq:ab-trace1}--\ref{eq:ab-trace3})} and get the desired bound stated in Equation~(\ref{eq:DimDiag}).

\partitle{Bounding the off-diagonal elements:}

From Lemma~\ref{lem:PairOfTiles} we get that
\begin{equation}
  2 \delta \tr(a \x b)
  \geq \abs{a_{r_1 r_2}} \cdot \abs{b_{c_1 c_2}}
  \label{eq:Offdiag12}
\end{equation}
for all $(r_1, c_1) \neq (r_2, c_2)$, except when $(r_1, c_1)$ and $(r_2, c_2)$ belong to the same tile of size two.

Suppose that we want to bound $\abs{a_{rr}} \cdot \abs{b_{c_1 c_2}}$ where $(r,c_1)$ and $(r,c_2)$ belong to the same $1 \times 2$ tile. Since $T$ is irreducible, we can find a row $r'$ such that $(r',c_1)$ and $(r',c_2)$ belong to different tiles (if $\set{r'} \times \set{c_1,c_2}$ is a tile for each $r'$ then $\set{c_1,c_2}$ is a connected component of the column graph of $T$, contradicting the irreducibility of $T$). From Equation~(\ref{eq:Offdiag2}) we get
\begin{equation}
  2 \delta \tr(a \x b)
  \geq \abs{a_{r'r'}} \cdot \abs{b_{c_1 c_2}}.
\end{equation}
According to Equation~(\ref{eq:a-diag}), $D \delta \tr(a) \geq \abs{a_{rr} - a_{r'r'}} \geq \abs{a_{rr}} - \abs{a_{r'r'}}$. Using this observation we repeat the calculation in Equations~\mbox{(\ref{eq:hard1}--\ref{eq:hard2})} and obtain
\begin{equation}
  2 \delta \tr(a \x b) \geq \abs{a_{rr}} \cdot \abs{b_{c_1 c_2}} - D \delta \tr(a \x b).
\end{equation}
After rearranging terms we get
\begin{equation}
  (D + 2) \delta \tr(a \x b) \geq \abs{a_{rr}} \cdot \abs{b_{c_1 c_2}}.
\end{equation}
The same bound also holds for entries corresponding to $2 \times 1$ tiles.
Together with Equation~(\ref{eq:Offdiag12}) this establishes the desired bound in Equation~(\ref{eq:DimOffdiag}).
\end{proof}

\section{Rigidity of rotated domino states (Lemma~\ref{lem:RotBox})} \label{apx:RotBox}

In this section we prove an analog of Lemma~\ref{lem:Box} for rotated domino states $S_3(\thetas)$. For simplicity we consider only the set $S_3(\theta) \defeq S_3(\theta,\theta,\theta,\theta)$ and obtain a bound as a function of $\theta$. In the more general case one can choose $\theta \defeq \min \set{\theta_1, \theta_2, \theta_3, \theta_4}$ and use the same bound.

\begin{lemma}\label{claim:a and b diagonal}
For $j \in \set{0,2}$ we have
\begin{equation}
  \abs{b_{11} - b_{jj}} 
  \leq \frac{2}{\sin 2 \theta}
       \bigl( \delta \norm{b}_{\infty} + \Abs{\Re b_{j1}} \bigr).
\end{equation}
The same inequality holds for $a$.
\end{lemma}

\begin{proof}
We show how to get the bound on $b$ for $j = 0$. The remaining three cases are similar.

We use the states $\ket{\psi_2}$ and $\ket{\psi_3}$ from Equation~(\ref{eq:psi23}) in the definition of $\delta$ in Equation~(\ref{eq:delta ab}):
\begin{align}
  \delta \norm{b}_{\infty}
   &\geq \abs{\bra{\beta_2} b \ket{\beta_3}} \\
   &=    \Abs{
           \mx{\cos \theta & \sin \theta}
           \mx{b_{00} & b_{01} \\ b_{10} & b_{11}}
           \mx{-\sin \theta \\ \cos \theta}
         } \\
   &=    \Abs{ (b_{11} - b_{00}) \sin \theta \cos \theta - b_{10} \sin^2 \theta + b_{01} \cos^2 \theta } \\
   &=    \Abs{ (b_{11} - b_{00}) \sin \theta \cos \theta
               + \Re b_{01} (\cos^2 \theta - \sin^2 \theta) + i \Im b_{01} } \\
   &\geq \Abs{ \frac{b_{11} - b_{00}}{2} \sin 2 \theta
               + \Re b_{01} \cos 2 \theta } \\
   &\geq \frac{\abs{b_{11}-b_{00}}}{2} \sin 2 \theta
               - \Abs{\Re b_{01}}.
\end{align}
By rearranging terms we get the desired bound.
\end{proof}

\begin{lemma}\label{claim:a and b bounds}
If $a_{11} \geq \frac{1}{s} \norm{a}_{\infty}$ for some $s > 0$ then for $j \in \set{0,2}$ we have
\begin{align}
  \abs{b_{j1}} &\leq \sqrt{2} s \delta \norm{b}_{\infty}, &
  \abs{b_{11} - b_{jj}} &\leq 2 (1 + \sqrt{2} s) \frac{\delta}{\sin 2 \theta} \norm{b}_{\infty}.
\end{align}
The same statement holds when the roles of $a$ and $b$ are exchanged.
\end{lemma}

\begin{proof}
We show how to get bounds on $b$ for $j = 0$. The remaining three cases are identical, except one has to use states from different tiles.

We use Lemma~\ref{lem:PairOfTiles} with tiles corresponding to states $\ket{\psi_{6,7}}$ and $\ket{\psi_1}$:
\begin{equation}
  \sqrt{2} \delta \norm{a \x b}_{\infty}
    \geq \abs{a_{11} b_{01}}
    \geq \frac{1}{s} \norm{a}_{\infty} \abs{b_{01}},
\end{equation}
where the second inequality follows from our assumption $\abs{a_{11}} \geq \frac{1}{s} \norm{a}_{\infty}$. By rewriting this we get the first bound:
\begin{equation}
  \abs{b_{01}} \leq \sqrt{2} s \delta \norm{b}_{\infty}.
  \label{eq:b01}
\end{equation}
Since $\Abs{\Re b_{01}} \leq \abs{b_{01}} \leq \sqrt{2} s \delta \norm{b}_{\infty}$, we get the second bound from Lemma~\ref{claim:a and b diagonal}.
\end{proof}

\begin{lemma}\label{claim:a and b max norm}
If $a_{11} \geq \frac{1}{s} \norm{a}_{\infty}$ and $b_{11} \geq \frac{1}{s} \norm{b}_{\infty}$ for some $s > 0$ then
\begin{equation}
  \norm{\frac{a \x b}{\tr(a \x b)} - \frac{\id}{9}}_{\max}
  \leq 8 (1 + \sqrt{2} s) \frac{\delta}{\sin 2 \theta}.
\end{equation}
\end{lemma}

\begin{proof}
We follow the proof of Lemma~\ref{lem:Box} and show the following generalizations of Equations~(\ref{eq:Diag}) and~(\ref{eq:Offdiag}):
\begin{align}
  \Abs{a_{ii} b_{jj} - \frac{1}{9} \tr(a \x b)}
    &\leq 8 (1 + \sqrt{2} s) \frac{\delta}{\sin 2 \theta} \norm{a \x b}_{\infty}, \label{eq:Rot-Diag} \\
  \abs{a_{ij}b_{kt}}
    &\leq \max \set{\sqrt{2}, 2, \sqrt{2} s} \delta \norm{a \x b}_{\infty}. \label{eq:Rot-Offdiag}
\end{align}
Note that the second inequality is stronger than we need, since $1 / \sin 2 \theta \geq 1$.

First, we use Lemma~\ref{claim:a and b bounds} to upper bound the difference of diagonal entries of $a$ and $b$. We use these bounds in the same way as in Lemma~\ref{lem:Box} to upper bound the differences of diagonal entries of $a \x b$ and to get Equation~(\ref{eq:Rot-Diag}). Finally, we use Lemma~\ref{lem:PairOfTiles} to upper bound most of the off-diagonal entries of $a \x b$ and Lemma~\ref{claim:a and b bounds} to upper bound the remaining ones. This gives us Equation~(\ref{eq:Rot-Offdiag}).

\partitle{Bounding the diagonal elements:}

From Lemma~\ref{claim:a and b bounds} we get bounds on $\abs{b_{11} - b_{ii}}$ and $\abs{a_{11} - a_{ii}}$ for $i \in \set{0,2}$. Using the triangle inequality, we get
\begin{equation}
  \abs{a_{ii} - a_{jj}} \leq 4 (1 + \sqrt{2} s) \frac{\delta}{\sin 2 \theta} \norm{a}_{\infty}
\end{equation}
for any $i,j \in \set{0,1,2}$ (and the same for $b$). Using the triangle inequality once more we can bound the difference of any two diagonal entries of $a \x b$:
\begin{equation}
  \abs{a_{ii} b_{jj} - a_{kk} b_{tt}}
  \leq 8 (1 + \sqrt{2} s) \frac{\delta}{\sin 2 \theta} \norm{a \x b}_{\infty}.
\end{equation}
From this we obtain Equation~(\ref{eq:Rot-Diag}) in the same way as in Lemma~\ref{lem:Box}.

\partitle{Bounding the off-diagonal elements:}

Equation~(\ref{eq:Rot-Offdiag}) can be obtained from Lemma~\ref{lem:PairOfTiles}. For most of the entries the constant is either $\sqrt{2}$ or $2$, depending on the sizes of the tiles. For the remaining four entries, listed in Equation~(\ref{eq:Hard cases}), we proceed in a slightly different way. For example, for $a_{00} b_{01}$ we use Equation~(\ref{eq:b01}) to see that
\begin{equation}
  \abs{a_{00}} \cdot \abs{b_{01}} \leq \norm{a}_{\infty} \cdot \sqrt{2} s \delta \norm{b}_{\infty}.
\end{equation}
A similar strategy works for the remaining three entries.
\end{proof}

\begin{lemma}
\label{claim:Small delta}
Fix any $s \geq 3$ and let
\begin{equation}
  \frac{1}{r(s)} \defeq \min \set{
                    \frac{1}{14} \biggl( \frac{1}{3} - \frac{1}{s} \biggr),
    \frac{1}{2 (1 + \sqrt{2} s)} \biggl( \frac{1}{3} - \frac{1}{s} \biggr)
  }.
  \label{eq:r}
\end{equation}
If $\frac{\delta}{\sin 2 \theta} \leq \frac{1}{r(s)}$ then $a_{11} \geq \frac{1}{s} \norm{a}_{\infty}$ and $b_{11} \geq \frac{1}{s} \norm{b}_{\infty}$.
\end{lemma}

\begin{proof}
We get one of the two lower bounds almost for free. We combine this with Lemma~\ref{lem:PairOfTiles} and the triangle inequality to get the other lower bound.

If $\max_i a_{ii} = a_{11}$ then $a_{11} \geq \frac{1}{3} \tr(a) \geq \frac{1}{3} \norm{a}_{\infty} \geq \frac{1}{s} \norm{a}_{\infty}$ and we are done with $a$. Similarly, if $\max_i b_{ii} = b_{11}$ then $b_{11} \geq \frac{1}{s} \norm{b}_{\infty}$. Thus it only remains to consider the cases when $\max_i a_{ii} \in \set{a_{00}, a_{22}}$ and $\max_i b_{ii} \in \set{b_{00}, b_{22}}$. By symmetry, it suffices to consider the case where $\max_i a_{ii} = a_{22}$ and $\max_i b_{ii} = b_{00}$. The remaining three cases are similar.

Using the tiles that correspond to states $\ket{\psi_{6,7}}$ and $\ket{\psi_{4,5}}$, we get
\begin{equation}
  2 \delta \norm{a \x b}_{\infty} \geq \abs{a_{22} b_{01}} \geq \frac{1}{3} \norm{a}_{\infty} \abs{b_{01}}.
\end{equation}
Thus $\Abs{\Re b_{01}} \leq \abs{b_{01}} \leq 6 \delta \norm{b}_{\infty}$ and using Lemma~\ref{claim:a and b diagonal}, we get
\begin{align}
  b_{00} - b_{11}
  & \leq \abs{b_{11} - b_{00}} \\
  & \leq \frac{2}{\sin 2 \theta} (\delta \norm{b}_{\infty} + \Abs{\Re b_{01}}) \\
  & \leq 14 \frac{\delta}{\sin 2 \theta} \norm{b}_{\infty}.
\end{align}
We assumed that $\max_i b_{ii} = b_{00}$, so
\begin{equation}
  \frac{1}{3} \norm{b}_{\infty} \leq b_{00} \leq b_{11} + 14 \frac{\delta}{\sin 2 \theta} \norm{b}_{\infty}.
\end{equation}
By assumption, $\frac{\delta}{\sin 2 \theta} \leq \frac{1}{r(s)} \leq \frac{1}{14} \bigl( \frac{1}{3} - \frac{1}{s} \bigr)$, so we get the desired bound $b_{11} \geq \frac{1}{s} \norm{b}_{\infty}$.

As we have a lower bound on $b_{11}$, we can use Lemma~\ref{claim:a and b bounds} and get
\begin{equation}
  \abs{a_{11} - a_{22}} \leq 2 (1 + \sqrt{2} s) \frac{\delta}{\sin 2 \theta} \norm{a}_{\infty}.
\end{equation}
We assumed that $\max_i a_{ii} = a_{22}$, so we can rewrite this as
\begin{equation}
  \frac{1}{3} \norm{a}_{\infty}
    \leq a_{22}
    \leq a_{11} + 2 (1 + \sqrt{2} s) \frac{\delta}{\sin 2 \theta} \norm{a}_{\infty}.
\end{equation}
By assumption, $\frac{\delta}{\sin 2 \theta} \leq \frac{1}{r(s)} \leq  \frac{1}{2 (1 + \sqrt{2} s)} \bigl( \frac{1}{3} - \frac{1}{s} \bigr)$, so we get the desired bound $a_{11} \geq \frac{1}{s} \norm{a}_{\infty}$.
\end{proof}

\begin{lemma}
\label{claim:Small and big delta}
For any fixed $s \geq 3$ we have the following:
\begin{itemize}
  \item if $\frac{\delta}{\sin 2 \theta} \leq \frac{1}{r(s)}$ then
        $\norm{\frac{a \x b}{\tr(a \x b)} - \frac{\id}{9}}_{\max}
         \leq 8 (1 + \sqrt{2} s) \frac{\delta}{\sin 2 \theta}$,
  \item if $\frac{\delta}{\sin 2 \theta} \geq \frac{1}{r(s)}$ then
        $\norm{\frac{a \x b}{\tr(a \x b)} - \frac{\id}{9}}_{\max}
         \leq r(s) \frac{\delta}{\sin 2 \theta}$,
\end{itemize}
where $r(s)$ is defined in Equation~(\ref{eq:r}).
\end{lemma}

\begin{proof}
The first part follows by combining Lemmas~\ref{claim:a and b max norm} and \ref{claim:Small delta}. To obtain the second part, notice that all diagonal entries of $\frac{a \x b}{\tr(a \x b)}$ are at most $1$. Since this matrix is positive semidefinite, the off-diagonal entries are also at most $1$, so the bound follows.
\end{proof}

\ROTBOX*

\begin{proof}
Let us denote the largest of the two constants in Lemma~\ref{claim:Small and big delta} by
\begin{align}
  C(s) \defeq& \max \set{8 (1 + \sqrt{2} s), r(s)} \\
            =& \max \set{8 (1 + \sqrt{2} s),
                             14 \frac{3s}{s-3},
             2 (1 + \sqrt{2} s) \frac{3s}{s-3}
       },
\end{align}
where we substituted $r(s)$ from Equation~(\ref{eq:r}). We want to make this constant as small as possible, so the best possible value is
\begin{align}
  C &=
  \min_{s \geq 3} \; C(s) \\
  &= \min_{s \geq 3} \; 2(1+\sqrt{2}s)\frac{3s}{s-3} \\
  &= 6 \Bigl( 1 + 6 \sqrt{2} + 2 \sqrt{3 (6 + \sqrt{2})} \Bigr),
\end{align}
where the minimum is reached at $s = 3 + \sqrt{9 + 3 / \sqrt{2}} \approx 6.33$.
\end{proof}

\end{document}